\documentclass[a4paper,UKenglish,cleveref, autoref, thm-restate]{lipics-v2021}

\hideLIPIcs  

\usepackage{pgffor}
\usepackage{pgfplots}
\usepackage{tikz,tikz-3dplot}
\usepackage{amsmath}
\usepackage{mathtools}
\usepackage[textsize=tiny, disable]{todonotes}
\usepackage{timetravel}
\usepackage{siunitx}

\setCounters{theorem}

\title{Computing Voronoi Diagrams in the Polar-Coordinate Model of the
  Hyperbolic Plane} 

\titlerunning{Voronoi Diagrams in the Polar-Coordinate Model of the
  Hyperbolic Plane} 

\author{Tobias Friedrich}{Hasso Plattner Institute, University of Potsdam\\{Potsdam, Germany}}{tobias.friedrich@hpi.de}{https://orcid.org/0000-0003-0076-6308}{}

\author{Maximilian Katzmann}{Karlsruhe Institute of Technology\\{Karlsruhe, Germany}}{maximilian.katzmann@kit.edu}{https://orcid.org/0000-0002-9302-5527}{}

\author{Leon Schiller}{Hasso Plattner Institute, University of Potsdam\\{Potsdam, Germany}}{leon.schiller@student.hpi.de}{}{}

\authorrunning{T. Friedrich, M.Katzmann, L. Schiller} 

\Copyright{Tobias Friedrich, Maximilian Katzmann, and Leon Schiller} 



\category{} 

\nolinenumbers 

\tikzset{
        partial ellipse/.style args={#1:#2:#3}{
            insert path={+ (#1:#3) arc (#1:#2:#3)}
        }
    }

\EventEditors{}
\EventNoEds{2}
\EventLongTitle{21st Symposium on Experimental Algorithms}
\EventShortTitle{SEA 2023}
\EventAcronym{SEA}
\EventYear{2023}
\EventDate{July 24--26, 2023}
\EventLocation{Barcelona, Spain}
\EventLogo{}
\SeriesVolume{40}
\ArticleNo{23}

\DeclareMathOperator{\acos}{acos}
\DeclareMathOperator{\asinh}{asinh}
\DeclareMathOperator{\acosh}{acosh}
\DeclareMathOperator{\atanh}{atanh}
\DeclareMathOperator{\atan}{atan}
\newcommand{\Vor}[1]{\mathrm{Vor}(#1)}
\newcommand{\Hplane}{\mathbb{H}^2}
\newcommand{\Set}[1]{\mathcal{#1}}
\newcommand{\Line}[2]{\overleftrightarrow{#1 #2}}
\newcommand{\Ray}[2]{\overrightarrow{#1 #2}}
\newcommand{\SingleRay}[1]{\overrightarrow{#1}}
\newcommand{\LineSegment}[2]{\overline{#1 #2}}
\newcommand{\Dist}[2]{|\LineSegment{#1}{#2}|}
\newcommand{\Bisector}[2]{#1\!\perp\!#2}
\newcommand{\Intersection}[2]{#1 \cap #2}
\newcommand{\SweepRadius}{\hat{r}}
\newcommand{\Active}[1]{\Set{A}(#1)}
\newcommand{\SitesInSweep}{\Set{S}_{\le\SweepRadius}}
\newcommand{\Triangle}[3]{\bigtriangleup #1 #2 #3}
\newcommand{\Angle}[3]{\sphericalangle #1 #2 #3}

\newcommand{\Ellipse}[1]{%
  \begin{tikzpicture}%
    \draw[rotate=66] (0, 0) ellipse (0.835ex and 0.45ex);%
    \draw[fill=black] (0.175ex, 0.45ex) circle (0.1ex);%
  \end{tikzpicture}%
  #1
}

\newcommand{\SmallEllipse}{%
  \!\begin{tikzpicture}%
    \draw[fill=white] (0.1, -0.075) circle (0.0ex);%
    \draw[rotate=66] (-0.15, -0.16) ellipse (0.575ex and 0.33ex);%
    \draw[fill=black] (0.098, -0.17) circle (0.075ex);%
  \end{tikzpicture}%
}

\newcommand{\EllipseFunction}[2]{r_{\SmallEllipse #1}(#2)}

\newcommand{\Circle}[1]{%
  \begin{tikzpicture}%
    \draw (0, 0) circle (0.805ex);%
    \draw[fill=black] (0, 0) circle (0.1ex);%
  \end{tikzpicture}%
  \hspace{0.75pt} #1
}

\newcommand{\CC}{C\nolinebreak\hspace{-.05em}\raisebox{.4ex}{\tiny\bf +}\nolinebreak\hspace{-.10em}\raisebox{.4ex}{\tiny\bf +}}
\def\CC{{C\nolinebreak[4]\hspace{-.05em}\raisebox{.4ex}{\tiny\bf ++}}}

\begin{document}

\maketitle

\begin{abstract}
  A Voronoi diagram is a basic geometric structure that partitions the
  space into regions associated with a given set of sites, such that
  all points in a region are closer to the corresponding site than to
  all other sites.
  While being thoroughly studied in Euclidean space, they are also of
  interest in hyperbolic space.
  In fact, there are several algorithms for computing hyperbolic
  Voronoi diagrams that work with the various models used to describe
  hyperbolic geometry.
  However, the \emph{polar-coordinate model} has not been considered
  before, despite its popularity in the network science community.
  While Voronoi diagrams have the potential to advance this field, the
  model is geometrically not as approachable as other models, which
  impedes the development of geometric algorithms.

  In this paper, we present an algorithm for computing Voronoi
  diagrams natively in the polar-coordinate model of the hyperbolic
  plane.
  The approach is based on Fortune's sweep line algorithm for
  Euclidean Voronoi diagrams.  We characterize the hyperbolic
  counterparts of the concepts it utilizes and introduce adaptations
  necessary to account for the differences.  
  
  We implemented our algorithm and compared it with the corresponding
  CGAL implementation.  While not being as numerically stable, our
  method has proven to be useful as a reference, which helped
  resolving fundamental issues in the implementation of the
  state-of-the-art method.
\end{abstract}

\newpage

\section{Introduction}

The Voronoi diagram is a fundamental and well-studied geometric
structure.  Given a set of points, which we call \emph{sites}, the
goal is to partition the space into \emph{cells}, i.e., regions that
are associated with the sites, such that no site is closer to a point
in a cell than the associated one. \todo{Unhide LIPICs for
  submission.} \todo{Include related version again.} \todo{Include
  line numbers again.}
Typically, the sites are assumed to lie in Euclidean space, where
finding the Voronoi diagram or its geometric dual, the Delaunay
triangulation, has various applications in biology, computer graphics,
and robotics~\cite{Bock_Tyagi_Kreft_Alt_2010, emp-tmpa-02,
  gma-ppmrn-06}.
However, this problem is also relevant in hyperbolic geometry, in the
context of finding certain Möbius transformations~\cite{be-omtiv-01},
computing Delaunay triangulations of points in two
planes~\cite{Boissonnat_Devillers_Teillaud_1996}, or for greedy
routing in networks~\cite{tim-rhvdt-11}.

Hyperbolic and Euclidean geometry differ in several fundamental
properties.  First, space expands exponentially fast in hyperbolic
space, while the expansion is only polynomial in Euclidean space.
Second, for a hyperbolic line there are infinitely many lines that go
through a point not on the line and are parallel to the first one.
And, third, in contrast to the flat Euclidean space, hyperbolic space
is negatively curved.

Over time, various models have been developed to capture these
properties and facilitate the study of hyperbolic geometry.  While
some models, like the \emph{Poincaré ball}, the \emph{Klein ball}, and
the \emph{hemisphere model} represent the infinite hyperbolic space in
a finite region in Euclidean space, other models like the
\emph{upper-half-plane model} and the \emph{hyperboloid model} map it
onto unbounded Euclidean regions.  For an overview we refer the reader
to~\cite[Chapter 7]{rr-ihg-95}.

With the different models came different approaches to computing
Voronoi diagrams.  For the upper-half-plane model there is an
algorithm that computes it from a Euclidean one~\cite{ot-cvdu-96}.  A
similar approach was later used to generalize an
$\mathcal{O}(n^2)$-algorithm in the Poincaré
disk~\cite{Nilforoushan_Mohades_2006}, to an
$\mathcal{O}(n \log(n))$-algorithm in the Poincaré
ball~\cite{Bogdanov_Devillers_Teillaud_2014}.  Another method is based
on power diagrams and utilizes the Klein disk
model~\cite{Nielsen_Nock_2010}. We refer the reader
to~\cite{tim-rhvdt-11} for an overview of hyperbolic Voronoi diagrams,
to~\cite{n-vdigc-20} for their relation to hyperbolic Delaunay
complexes, and to~\cite{nn-vhfd-14} for visualizations in different
models.  We note that converting between the various representations
typically involves scaling the coordinates exponentially, which can
lead to numerical issues.  Consequently, it is best to perform all
computations directly in one model, instead of taking detours through
other representations.

A model that has not been considered in the context of hyperbolic
Voronoi diagrams before is the \emph{polar-coordinate model}, which is
rather surprising given its popularity in the study of complex
networks.
There, the polar-coordinate representation is used to analyze
real-world networks like protein-interaction networks, the internet,
and trade networks~\cite{ama-lghpi-18,
  Boguna_Papadopoulos_Krioukov_2010, gbas-whhgit-16}.  In particular,
\emph{hyperbolic random graphs}, a generative graph model aimed at
representing such networks, were introduced using the polar-coordinate
model of the hyperbolic plane~\cite{kpf-hgcn-10}, and it was since
shown that this representation is particularly well-suited for
mathematical analysis of network properties and
algorithms~\cite{bff-svcpt-21, bff-espsn-18,
  Gugelmann_Panagiotou_Peter_2012, k-sahrg-16, ms-k-19}.  Being able
to compute Voronoi diagrams in this model has the potential to further
advance this field.

One reason for the mathematical accessibility of the polar-coordinate
model is its straightforward definition.  Points in hyperbolic space
are addressed using their distance to a dedicated point, called
\emph{pole}, and the angular distance to a ray starting at the pole,
called \emph{polar axis}.  The resulting coordinates are then simply
interpreted as polar coordinates in Euclidean space.
However, in contrast to models like the Poincaré disk, where
hyperbolic circles are Euclidean circles (with offset centers) and
hyperbolic lines are circular arcs in Euclidean space, or the
hyperboloid model, where hyperbolic circles and lines are
intersections of planes with the hyperboloid, the polar-coordinate
model is not as approachable from a geometric point of view.  There,
hyperbolic circles are shaped like tear drops and lines are hyperbolas
that are bent towards the pole (see Figure~\ref{fig:model}~(left) for
an illustration).  As a consequence, this model has not been
considered in the context of geometric algorithms before.

In this paper, we present an algorithm for computing Voronoi diagrams
directly in the polar-coordinate model of the two-dimensional
hyperbolic plane.
The approach is based on Fortune's sweep line algorithm for solving
the problem in the Euclidean plane~\cite{Fortune_1986}, which scans
the plane and maintains data structures to incrementally build the
diagram using a so-called \emph{beach line}.  More precisely, our
approach builds on a generalization of Fortune's algorithm that uses
an expanding sweep circle instead of a
line~\cite{Xin_Wang_Xia_Mueller_Wittig_Wang_He_2013}, while
maintaining an optimal running time of $\mathcal{O}(n \log(n))$.

The resulting algorithm is simple and while proving its correctness
follows known techniques, we note that translating them to work with
the considered model of hyperbolic geometry is not trivial.  In
contrast, prior work nicely demonstrates that other models allow for
utilizing parallels to Euclidean geometry (see,
e.g.,~\cite{Bogdanov_Devillers_Teillaud_2014, ot-cvdu-96}), yet the
resulting algorithms are more complicated and thus harder to
implement.


We implemented our algorithm and compared it to the state-of-the-art
method for computing hyperbolic Voronoi diagrams and Delaunay
complexes~\cite{Bogdanov_Devillers_Teillaud_2014}, which is available
in CGAL~\cite{bit-hdt-22}.  Since this method works with the Poincaré
disk model of the hyperbolic plane, utilizing it for our purposes
requires converting between the Poincaré disk and the polar-coordinate
model.  Surprisingly, our results show that the computation of the
Voronoi vertices is more reliable than in our native approach, despite
these conversions.  However, further comparisons with our method
showed that the CGAL implementation was not able to reliably compute
the correct structure of the diagram.  In fact, adjacent Voronoi cells
were often not detected as such (corresponding to missing edges in the
Delaunay triangulation).  After reporting the
issue\footnote{\url{https://github.com/CGAL/cgal/issues/6869}}, it
turned out that solving it required a sizable change in the
implementation\footnote{\url{https://github.com/MaelRL/cgal/commit/cf12f90cbf721fbbf6752d98c6d76190b6a5052c}},
which highlights the usefulness of reference implementations that are
based on simple algorithms.  Beyond that, we observe that both methods
suffer from numerical inaccuracies.  Thus, while our method represents
the first step towards utilizing geometric algorithms in the
polar-coordinate model of the hyperbolic plane, further research is
necessary to obtain scalable methods.


\todo{Adjust for submission} We note that proofs and in particular the
proof of correctness and running time analysis of our algorithm are
deferred to Appendix~\ref{apx:missing-proofs}.


\section{Preliminaries}
\label{sec:preliminaries}

\begin{figure}[t]
  \centering
  \includegraphics{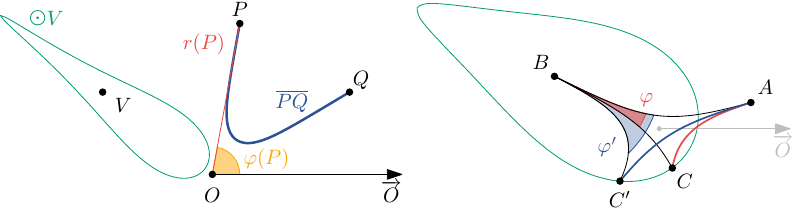}
  \caption{\textbf{(Left)} The polar-coordinate model of the
    hyperbolic plane.  The line segment $\LineSegment{P}{Q}$ is shown
    in blue.  Point $V$ is the center of the tear-drop shaped green
    circle $\Circle{V}$.  \textbf{(Right)} Illustration of the
    triangles $\Triangle{A}{B}{C}$ and $\Triangle{A}{B}{C'}$ from
    Lemma~\ref{lem:angle-affects-distance}.  Since $C$ and $C'$ lie on
    a circle (green) with center $B$, we have
    $\Dist{B}{C} = \Dist{B}{C'}$.  When increasing the angle $\varphi$
    at $B$ to $\varphi'$ then $\Dist{A}{C} < \Dist{A}{C'}$.}
  \label{fig:model}
\end{figure}

We denote points with capital letters like $P, Q$ and sets or tuples
with calligraphic capital letters like $\Set{S}, \Set{V}$.  Given two
points $P, Q$, we denote the line through them with $\Line{P}{Q}$ and
the ray from $P$ through $Q$ with $\Ray{P}{Q}$.  A ray starting at
$P$, whose direction is given by context is denoted by
$\SingleRay{P}$.  The line segment between $P$ and $Q$ is written as
$\LineSegment{P}{Q}$ and $\Dist{P}{Q}$ denotes its length, i.e., the
distance between $P$ and $Q$.  The \emph{perpendicular bisector},
i.e., the set of points with equal distance to $P$ and $Q$ is denoted
by $\Bisector{P}{Q}$.  A \emph{circle} with center $P$, i.e., the set
of all points $Q$ with equal distance to $P$ is denoted by
$\Circle{P}$ (its radius is given by context).  The points $Q$ are
said to lie on the \emph{arc} of $\Circle{P}$.  Triangles are written
as $\Triangle{P}{Q}{S}$ and $\Angle{P}{Q}{S}$ denotes the angle
between $\Ray{Q}{P}$ and $\Ray{Q}{S}$ measured in \emph{clockwise}
direction around~$Q$.  The intersection of two objects is denoted with
a $\Intersection{}{}$-sign.

\subparagraph{The Hyperbolic Plane.}  In this paper, we work with the
\emph{polar-coordinate model} of the hyperbolic plane $\Hplane$.
There, points are identified using polar coordinates.  After defining
a designated \emph{origin} or \emph{pole} $O \in \Hplane$ together
with a \emph{polar axis} $\SingleRay{O}$, i.e., a reference ray
starting at $O$, a point $P \in \Hplane$ is identified using its
\emph{radius} $r(P)$ denoting the hyperbolic distance to $O$ and its
\emph{angle} $\varphi(P)$ denoting the angular distance to
$\SingleRay{O}$ in counterclockwise direction around $O$.  For
visualizations, these coordinates are then interpreted as polar
coordinates in the Euclidean plane, as shown in
Figure~\Ref{fig:model}.  We note that angles at line intersections are
\emph{not} preserved by the representation and are only added in our
figures for illustration purposes.

In contrast to Euclidean geometry, the sum of the angles in a triangle
$\Triangle{A}{B}{C}$ is strictly less than $\pi$ in the hyperbolic
plane.  Still, we can relate the length of $\LineSegment{A}{B}$ to the
lengths of the other segments and the angle $\varphi$ at $C$, via the
\emph{hyperbolic law of cosines}, which is defined as
\begin{align}
  \label{eq:law-of-cosines}
  \cosh(\Dist{A}{B}) = \cosh(\Dist{A}{C}) \cosh(\Dist{B}{C}) - \sinh(\Dist{A}{C}) \sinh(\Dist{B}{C}) \cos(\varphi),
\end{align}
where $\sinh(x) = (e^{x} - e^{-x}) / 2$,
$\cosh(x) = (e^{x} + e^{-x}) / 2$.  Analogous to Euclidean geometry,
the hyperbolic tangent $\tanh$ is defined as the quotient of the
hyperbolic sine and cosine.  We denote the inverse hyperbolic
trigonometric functions with $\asinh$, $\acosh$, and $\atanh$.  Note
that we can use the hyperbolic law of cosines to derive the distance
between $A, B \in \Hplane$ by considering the triangle
$\Triangle{A}{O}{B}$, where the angle at $O$ is
$\varphi = \Delta_{\varphi}(A, B) = \pi - |\pi - |\varphi(A) -
\varphi(B)||$.  Since $\cos(\pi - x) = -\cos(x)$ and
$\cos(-x) = \cos(x)$, we have
$\cos(\Delta_{\varphi}(A, B)) = \cos(\varphi(A) - \varphi(B))$.  Thus,
the hyperbolic distance between $A$ and $B$ is given by
\begin{align}
  \label{eq:hyperbolic-distance}
  \Dist{A}{B} &= \acosh \big( \cosh(r(A)) \cosh(r(B)) - \sinh(r(A)) \sinh(r(B)) \cos(\varphi(A) - \varphi(B)) \big).
\end{align}
The distance between a circle of radius $r$ centered at the pole and a
point $A$ with radius $r(A) \le r$ is given by $r - r(A)$.  The
following lemma shows that changing an angle in a triangle also
changes the length of the opposite line segment, as illustrated in
Figure~\ref{fig:model}~(right).

\wormhole{lem-angle-affects-distance}
\begin{lemma}
  \label{lem:angle-affects-distance}
  Let $\Triangle{A}{B}{C}$ and $\Triangle{A}{B}{C'}$ be triangles with
  $\Dist{B}{C} = \Dist{B}{C'}$ and angles $\varphi$ and~$\varphi'$ at
  $B$, respectively.  Then, $\varphi\!<\!\varphi'$ (resp.
  $\varphi\!>\!\varphi'$) if and only if $\Dist{A}{C} < \Dist{A}{C'}$
  (resp.  $\Dist{A}{C} > \Dist{A}{C'}$).
\end{lemma}
We note that in a degenerate triangle where the angle at $B$ is $\pi$,
the length of the segment opposite of $B$ is given by the sum of the
lengths of the other two segments.  Thus, the above lemma confirms the
fact that the triangle inequality also holds in the hyperbolic plane.
The following lemma shows that, if two points $B, C$ have certain
distances to a third point $A$, then every distance between the two
distances is realized by a point on $\LineSegment{B}{C}$.

\wormhole{lem-distance-continuous}
\begin{lemma}
  \label{lem:distance-continuous}
  Let $\Triangle{A}{B}{C}$ be a triangle with $\Dist{A}{B} \le
  \Dist{A}{C}$.  For every $d \in [\Dist{A}{B}, \Dist{A}{C}]$, there
  exists a point $P \in \LineSegment{B}{C}$ with $\Dist{A}{P} = d$.
\end{lemma}

\subsection{Hyperbolic Voronoi Diagram}
\label{sec:prelim-voronoi}

The hyperbolic Voronoi diagram is defined analogously to the Euclidean
version~\cite[Chapter~7]{bcko-cg-08}.  Let~$\Set{S}$ be a set of
finitely many points in $\Hplane$.  For each \emph{site}
$S \in \Set{S}$ we define the \emph{Voronoi cell} $\Vor{S}$ as the set
of points $P \in \Hplane$ such that no site in $\Set{S}$ is closer to
$P$ than~$S$.  The \emph{Voronoi diagram} $\Vor{\Set{S}}$ is the
partition of $\Hplane$ into the Voronoi cells.  The diagram is
formalized by the set $\Set{E}$ of \emph{Voronoi edges} denoting the
borders between cells, together with the set $\Set{V}$ of
\emph{Voronoi vertices} denoting the endpoints of the edges.  Note
that, by definition, the Voronoi edge between two sites $S_1 \neq S_2$
is a subset of the bisector $\Bisector{S_1}{S_2}$.  Further note that
$\Set{V}$ contains all points $V \in \Hplane$ that have three closest
sites, i.e., $V \in \Set{V}$ if and only if it is the center of an
empty circle $\Circle{V}$ with at least three sites
$S_1, S_2, S_3, \ldots \in \Set{S}$ on its
arc~\cite{Bogdanov_Devillers_Teillaud_2014}.  See
Figure~\ref{fig:hyperbolic-voronoi-prelims} for an illustration.  We
call this circle the \emph{witness circle of $V$}.  In the context of
hyperbolic Voronoi diagrams, we are often interested in the point
$V' \in \Circle{V}$, denoted as the \emph{far point of $V$}, that is
farthest from the pole, i.e., the point with the largest radial
coordinate among points on the witness circle of $V$.  We call the
sites $S_1, S_2, S_3, \dots$ the \emph{incident sites} of $V$ and
refer to the tuple containing the incident sites in the order one
encounters them when traversing $\Circle{V}$ in clockwise direction
starting at the far point $V'$, as the \emph{incidence tuple of $V$}.

\begin{figure}[t]
  \centering
  \includegraphics{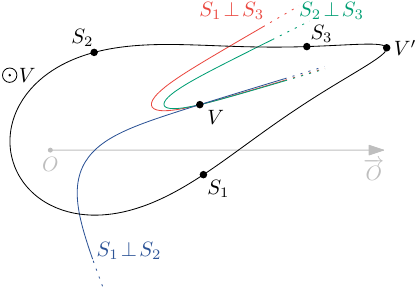}
  \caption{A point $V$ with three sites $S_1, S_2$, and $S_3$ and the
    far point $V'$ on the arc of its witness circle $\Circle{V}$.  The
    perpendicular bisectors of the sites intersect at $V$.}
  \label{fig:hyperbolic-voronoi-prelims}
\end{figure}

\subsection{Fortune's Algorithm}
\label{sec:fortunes-algorithm}

\subparagraph{Euclidean Sweep Line.} Given a set $\Set{S}$ of sites in
the Euclidean plane, Fortune's algorithm~\cite{Fortune_1986} computes
the Voronoi diagram $\Vor{\Set{S}}$ using a \emph{sweep line} that can
be thought of as being parallel to the $x$-axis and traversing the
plane in negative $y$-direction (Figure~\ref{fig:foruntes-algorithm}
(left)).  At any time all sites above the sweep line are incorporated
into the diagram.  The sites below the sweep line have not been seen,
yet, and may affect the existing diagram, at least within a certain
region above the sweep line.  The boundary that separates this region
from the completed part of the diagram is called the \emph{beach
  line}.  It consists of parts of \emph{beach parabolas}, which
contain the points that lie at equal distance between a site and the
sweep line.  The intersections of neighboring beach parabolas trace
the edges of the Voronoi diagram.  Maintaining the beach line is the
crucial part of the algorithm.  Of course, one cannot perform a
continuous sweep motion in a computer program.  Instead the algorithm
makes use of the fact that the beach line only changes at certain
sweep line positions, called \emph{site events} and \emph{circle
  events}.  Site events occur when the sweep line reaches a site and a
new beach parabola is added to the beach line.  A circle event occurs
when the sweep line reaches the lowest point of a circle containing
three sites whose beach parabolas are consecutive on the beach line.
There a parabola is potentially removed from the beach line and a
Voronoi vertex is detected.  For a more detailed description, we refer
the reader to~\cite[Section 7.2]{bcko-cg-08}.

\begin{figure}[t]
  \centering
  \includegraphics{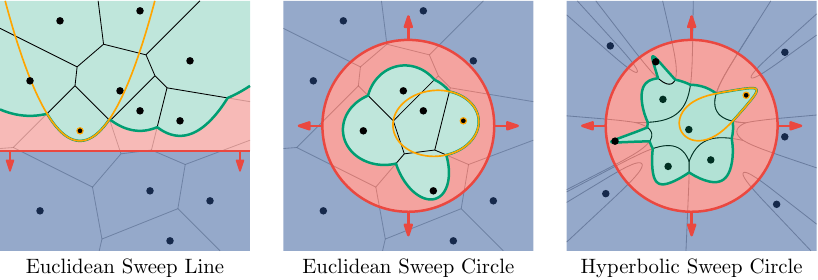}
  \caption{Illustration of Fortune's algorithm and its adaptations.
    Black dots denote the sites.  The sweep line (or circle) is dark
    red; arrows indicate its movement.  The green region denotes the
    completed part of the diagram (black edges).  The orange curve is
    the beach parabola (or ellipse) of the site with orange border.
    The beach line (or curve) is dark green.  The red region denotes
    the part between sweep line (or circle) and beach line (or curve).
    Blue areas denote the unseen region.}
  \label{fig:foruntes-algorithm}
\end{figure}

\subparagraph{Euclidean Sweep Circle.} It was later shown that
Fortune's algorithm is actually a degenerate form of a \emph{sweep
  circle} algorithm~\cite[Theorem
1]{Xin_Wang_Xia_Mueller_Wittig_Wang_He_2013}.  Instead of a sweep line
that traverses the plane, the idea is to use a \emph{sweep circle}
that grows from a certain point (see
Figure~\ref{fig:foruntes-algorithm} (center)).  Then, beach parabolas
become \emph{beach ellipses} that contain the points at equal distance
between a site and the sweep circle, and the beach line becomes a
\emph{beach curve} consisting of beach ellipse segments.  If one
imagined the center of the sweep circle to lie at infinite distance to
the sites, one obtains the original sweep line algorithm.

\section{A Hyperbolic Sweep Circle Approach}
\label{sec:hyperbolic-sweep-circle-approach}

In the polar-coordinate model, ``straight'' lines are hyperbolas that
are bent towards the pole.  Consequently, sweeping a line through the
hyperbolic plane is rather tedious, as even the scheduling of site
events leads to difficult computations.  A much more natural approach
is obtained by considering a sweep circle of expanding radius centered
at the pole.  This presents the first difference to the Euclidean
approach: While the choice for the center of the sweep circle is
basically irrelevant in Euclidean space, the pole is the only
reasonable choice in the hyperbolic plane, as the occurrence of site
events can be easily read from the radii of the sites then.
Scheduling circle events and handling events is more involved.  As
shown in Figure~\ref{fig:foruntes-algorithm} (right), the hyperbolic
counterparts of beach ellipses and the beach curve are rather
different from their Euclidean versions.  Nevertheless, the hyperbolic
sweep circle approach works analogously to the Euclidean variant.  The
idea is to simulate the expansion of the sweep circle, whose
increasing radius we denote with $\SweepRadius$, and to maintain two
data structures, which we describe in the following.

\subsection{Beach Curve}
\label{sec:beach-curve}

In the Euclidean sweep line approach the beach line consists of parts
of beach parabolas.  For each site $S \in \Set{S}$ above the sweep
line, the beach parabola represents all points with equal distance to
$S$ and the sweep line.  In the Euclidean sweep circle method, the
parabolas become ellipses.  They are defined for sites in the sweep
circle, denoted by
$\SitesInSweep = \{S \in \Set{S} \mid r(S) \le \SweepRadius\}$, and
consist of all points at equal distance to a site and the sweep
circle.

In the hyperbolic plane the beach ellipse of $S$, denoted by
$\Ellipse{S}$, is not actually elliptic
(Figure~\ref{fig:foruntes-algorithm} (right)).  However, we can
parameterize $\Ellipse{S}$ by the angles of the points that lie on it.
To this end, we define a function $\EllipseFunction{S}{\varphi}$ that
maps an angle $\varphi \in [0, 2\pi)$ to the radius of the point
$P \in \Ellipse{S}$ with $\varphi(P) = \varphi$.

\wormhole{lem-ellipse-curve}
\begin{lemma}
  \label{lem:ellipse-curve}
  Let $\SweepRadius > 0$ be the radius of the sweep circle, let $S \in
  \SitesInSweep$ be a site with $r(S) < \SweepRadius$, and let $P \in
  \Hplane$ be a point with $\varphi(P) = \varphi$.  Then, $P \in
  \Ellipse{S}$ if and only if $P$ has radius
  \begin{align*}
    \EllipseFunction{S}{\varphi} = \atanh \left(\frac{\cosh(\SweepRadius) - \cosh(r(S))}{\sinh(\SweepRadius) - \sinh(r(S))\cos(\varphi - \varphi(S))}\right ).
  \end{align*}
\end{lemma}
When the site $S$ is not inside the sweep circle, but on its arc
(i.e., $r(S) = \SweepRadius$), then its ellipse degenerates into the
line segment $\Ellipse{S} = \LineSegment{O}{S}$.  The beach curve
$\Set{B}$ is defined as the set of points lying on the outer most
parts of the beach ellipses, as shown in
Figure~\ref{fig:beach-structure}~(left).  Formally, we define the
\emph{active segments} of a site $S \in \SitesInSweep$ as
$\Active{\Ellipse{S}} = \{P \in \Ellipse{S} \mid \nexists S' \in
\SitesInSweep, P' \in \Ellipse{S'} \colon \varphi(P) = \varphi(P')
\land r(P) < r(P')\}$.  For an angular coordinate
$\varphi \in [0, 2\pi)$, we say that a site $S$ is \emph{active at
  $\varphi$}, if there exists a point $P \in \Active{\Ellipse{S}}$
with $\varphi(P) = \varphi$.  The beach curve~$\Set{B}$ is then given
by the union of the active segments of sites in the sweep circle,
i.e., $\Set{B} = \bigcup_{S \in \SitesInSweep} \Active{\Ellipse{S}}$.

Now consider two sites $S, T \in \SitesInSweep$ together with a point
$P$ lying on an intersection of their beach ellipses.  The distance
between $P$ and the sweep circle is given by $\SweepRadius - r(P)$ and
is, by definition of the beach ellipse, equal to $\Dist{P}{S}$ and to
$\Dist{P}{T}$.  We obtain the following.

\begin{observation}
  Let $S, T \in \SitesInSweep$ be two sites with beach ellipses
  $\Ellipse{S}$, $\Ellipse{T}$, and let
  $P \in \Intersection{\Ellipse{S}}{\Ellipse{T}}$ be a point on an
  intersection of them.  Then, $P$ lies on the perpendicular bisector
  $\Bisector{S}{T}$.
\end{observation}

Thus, intersections of active beach ellipse segments move along the
bisectors of the sites as the sweep circle expands, and trace the
Voronoi edges, which are subsets of these bisectors.

\begin{figure}[t]
  \centering
  \includegraphics{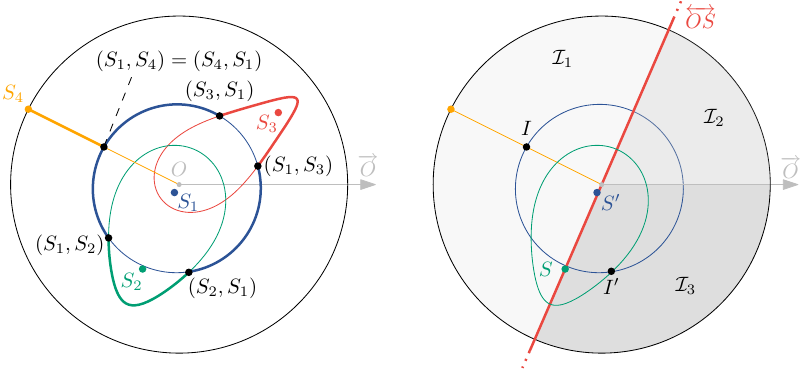}
  \caption{\textbf{(Left)} Four sites with beach ellipses.
    Since $S_4$ (orange) lies on the sweep circle (black),
    its beach ellipse is degenerate.  Active segments are bold.
    Intersections are black points.
    \textbf{(Right)} We implicitly compare $\varphi(I)$ and
    $\varphi(I')$ using the gray sectors $\Set{I}_1, \Set{I}_2$, and
    $\Set{I}_3$ that are defined using~$\Line{O}{S}$.}
  \label{fig:beach-structure}
\end{figure}

We represent the beach curve $\Set{B}$ as a tuple $\Set{I}$ of
intersections of active segments that are ordered by their angular
coordinates.  Since the beach curve is a closed curve, start and end
of the tuple are identified.  Note that the angular coordinates of the
intersections are constantly changing as the sweep circle expands.
Therefore, we represent the intersections implicitly as a tuple of the
two intersecting sites.  More precisely, consider an active segment of
a site $S$, that is bounded by the intersections $I$ and $I'$ with two
other active segments belonging to the sites $T$ and $T'$,
respectively.  If $I$ and $I'$ appear in this order when traversing
the beach curve in counterclockwise direction, then the segment is
represented by the corresponding tuples $(T, S)$ and $(S, T')$ in this
order.  See Figure~\ref{fig:beach-structure}~(left) for an
illustration of intersection tuples.

As the sweep circle expands, intersections enter and leave the beach
curve $\Set{B}$, so $\Set{I}$ needs to be maintained accordingly,
without knowing the angular coordinates of the intersections.  For
efficient insertions and deletions, we perform a binary search, while
potentially reducing the number of intersections whose actual
coordinates need to be computed.

Consider a new intersection $I$ that enters the beach curve and assume
that its angular coordinate $\varphi(I)$ is given.
In order to find its position in the data structure $\Set{I}$ we want
to perform a binary search and thus need to be able to decide whether
$\varphi(I) < \varphi(I')$ for a given intersection $I' \in \Set{I}$,
where $\varphi(I')$ is not known.  Let $S, S' \in \SitesInSweep$ be
the two sites with $I' \in \Intersection{\Ellipse{S}}{\Ellipse{S'}}$
and assume without loss of generality that $r(S) \ge r(S')$.  The idea
now is to use the line $\Line{O}{S}$ to decide whether
$\varphi(I) < \varphi(I')$.  If $\Ellipse{S}$ is degenerate, then
$\varphi(I') = \varphi(S)$ and the decision is straightforward.
Otherwise, we split the angular interval $[0, 2\pi)$ into three parts.
Interval $\Set{I}_1$ contains the angles of all points lying on the
opposite side of $\Line{O}{S}$ as the polar axis~$\SingleRay{O}$.
Intervals $\Set{I}_2$ and $\Set{I}_3$ contain the angles on the same
side of $\Line{O}{S}$ as $\SingleRay{O}$ that are smaller and larger
than $\varphi(S)$, respectively.  See
Figure~\Ref{fig:beach-structure}~(right) for an illustration of the
intervals.  Then, if $\varphi(I)$ and $\varphi(I')$ are in different
intervals, it is easy to decide whether $\varphi(I) < \varphi(I')$.
Only if they are in the same interval we need to compute the angular
coordinate of $I'$.  Thus, it may suffice to determine the intervals
that $\varphi(I)$ and $\varphi(I')$ lie in, without
computing~$\varphi(I')$ explicitly.  Since $\varphi(I)$ is given,
determining its interval is trivial.  Finding the one containing
$\varphi(I')$ is more involved.  We can compute this information once
when $I'$ enters the beach curve, since we know its angular coordinate
at this moment.  Then, if $I' \in \Set{I}_2$, this does not change as
the sweep circle expands.
However, if $I' \in \Set{I}_1$ or $I' \in \Set{I}_3$, it may move from
one interval to another by passing angular coordinate $0$ as the sweep
circle expands.  Then $I'$ moves from the first position in the beach
curve data structure $\Set{I}$ to the last or vice versa, which we
call a \emph{structure event}.  Such events are scheduled or canceled
every time the first or last element in $\Set{I}$ changes.

In conclusion, we now know how to maintain $\Set{I}$ and that we can
use the contained intersections to trace the Voronoi edges.  It
remains to determine the events at which intersections enter and leave
$\Set{I}$ and how we can use them to detect the Voronoi vertices.

\subsection{Event Queue}
\label{sec:event-queue}

The second data structure the algorithm maintains is a priority queue
called \emph{event queue} $\Set{Q}$, which stores the events at which
the beach curve changes, in the order an expanding sweep circle
encounters them.  That is, each event is associated with a point $P$
and its priority is~$r(P)$.  Analogous to the Euclidean version (see
Section~\ref{sec:fortunes-algorithm}), there are two types of events
that change the beach curve by either adding or removing active beach
ellipse segments.

An active beach ellipse segment is added to the beach curve when the
sweep circle expands beyond a previously unseen site, i.e., when the
radius $\SweepRadius$ of the sweep circle is equal to the radial
coordinate of that site.  (See site $S_4$ in
Figure~\ref{fig:beach-structure}~(left) for an example).  When
encountering such a \emph{site event}, the intersections of the newly
added active segment with an existing active segment are added to the
tuple $\Set{I}$ and the bisectors between the new site and the
corresponding existing site are marked as part of the Voronoi diagram.
Clearly, all site events can be scheduled in advance, since the radial
coordinates of the sites are given.

\begin{figure}[t]
  \centering
  \includegraphics{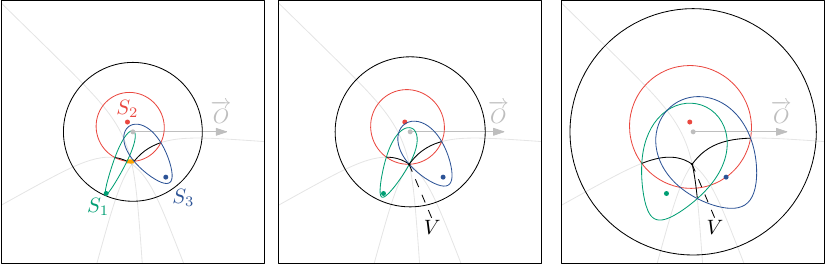}
  \caption{Beach curve intersections within an expanding sweep circle
    (black).  \textbf{(Left)} An active segment (orange) of $S_2$
    (red) is about to vanish.  \textbf{(Center)} The active segment
    has vanished as the intersections of $\Ellipse{S_1}$ and
    $\Ellipse{S_3}$ with $\Ellipse{S_2}$ meet at a point $V$.
    \textbf{(Right)} The vanished segment is replaced by a single
    intersection of $\Ellipse{S_1}$ and $\Ellipse{S_3}$.}
  \label{fig:merge}
\end{figure}

On the other hand, an active segment can be removed from the beach
curve, which happens when its two intersections with other active
segments merge, as shown in Figure~\ref{fig:merge}.  Let
$S_1, S_2, S_3 \in \SitesInSweep$ be three sites such that an active
segment of $\Ellipse{S_2}$ intersects an active segment of
$\Ellipse{S_1}$ and one of $\Ellipse{S_3}$.  When the two
intersections merge at a point~$V$, then $V$ has equal distance to the
sweep circle and to the three sites $S_1, S_2$, and $S_3$.
Consequently, all of them are incident to a circle $\Circle{V}$ that
is completely contained in the sweep circle.  Such an event is,
therefore, called \emph{circle event}.  Note that no unseen site can
lie in $\Circle{V}$.  Thus, if no site is contained in $\Circle{V}$ as
the circle event occurs, then $V$ is the center of an empty circle
and, therefore, a Voronoi vertex (see
Section~\ref{sec:prelim-voronoi}).  Then, $V$ is added to the diagram
and the bisectors of the incident sites are marked as being incident
to~$V$.  Moreover, the tuple of active beach ellipse intersections
$\Set{I}$ is adjusted, by removing the merged intersections and
replacing them with a new intersection of the corresponding active
segments of $\Ellipse{S_1}$ and $\Ellipse{S_3}$.

Such a circle event occurs when the sweep circle reaches the far point
$V'$ of $V$, i.e., the point with the largest radial coordinate among
all points on $\Circle{V}$.  Thus, whenever an active segment enters
or leaves the beach curve, we can use its intersections in $\Set{I}$
to determine the neighboring active segments and schedule circle
events by determining the point $V'$ for the circle the corresponding
sites are incident to.  We note that not all tuples of consecutive
beach curve intersections lead to a circle event, as the corresponding
bisectors may not intersect or the beach ellipse intersections may
diverge as the sweep circle expands.  Consequently, we distinguish
between \emph{true} and \emph{false} circle events and handle them
accordingly.

Instead of a continuous sweep motion, the algorithm then computes the
Voronoi diagram by iteratively processing the events in the event
queue, until the queue is empty.

Following the proof of correctness and the running time analysis of
the Euclidean version of the algorithm (see e.g.~\cite[Section
7.2]{bcko-cg-08}) while taking our adaptations into account, we can
show that the above sketched algorithm correctly computes the
hyperbolic Voronoi diagram.

\wormhole{thm-main}
\begin{theorem}
  \label{thm:main}
  Let $\Set{S} = \{S_1, \dots, S_n\} \subset \Hplane$ be a set of
  sites.  Then, the sweep circle algorithm computes $\Vor{\Set{S}}$ in
  time $\mathcal{O}(n \log(n))$.
\end{theorem}

\section{Experiments}
\label{sec:experiments}

We implemented our
algorithm\footnote{\url{https://github.com/maxkatzmann/fortune-hyperbolic}}
in \CC{} and compared it to the existing state-of-the-art
implementation for computing two-dimensional hyperbolic Delaunay
complexes in CGAL~\cite{bit-hdt-22}.
Figure~\ref{fig:computed-drawings} shows exemplary Voronoi diagrams
and Delaunay complexes computed using our implementation.

\begin{figure}[t]
  \centering
  \includegraphics{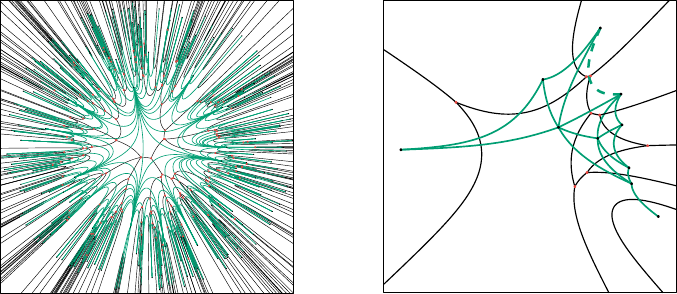}
  \caption{Voronoi diagrams for different sets of sites (black)
    computed using our approach, with Voronoi vertices in red and
    edges in black.  The Delaunay complex is shown in green.
    \textbf{(Left)} 227 sites are distributed in a disk of radius
    $\num{8.315}$. \textbf{(Right)} 10 sites are distributed in a disk
    of radius $\num{3.171}$.  The dashed edge was not found when using
    CGAL (Converted) before the fix.}
  \label{fig:computed-drawings}
\end{figure}

Since computations in hyperbolic space are notoriously prone to
numerical difficulties, the goal of the experiments is to determine
how long the different approaches yield viable results as the size of
the space containing the sites increases.  To this end, we considered
hyperbolic disks of increasing radii centered at the pole in the
polar-coordinate model of the hyperbolic plane and distributed sites
uniformly at random in them, according to the hyperbolic metric.
Recall that the area of a hyperbolic disk of radius~$r$ grows as
$2\pi (\cosh(r) - 1)$.  By setting the number of sites in a disk of
radius $r$ to $N(r) = 2 \pi (\cosh(c \cdot r) - 1)$ for a constant
$c$, this number is roughly proportional to the disk area, leading to
equally densely filled disks for different radii.  To account for the
exponential expansion of space, we then chose 10 radii in
logarithmically increasing steps, up to a maximum radius of $20$, and
chose $c$ such that $N(20) = \num{100000}$.  To obtain statistically
significant results, we performed $\num{100}$ samples for each disk,
leading to a data set of $\num{1000}$ instances.  As we explain below,
all experiments were, in fact, performed twice.

For each instance, we then computed the Voronoi diagram and
corresponding Delaunay complex using different techniques.  The first
technique is the existing state-of-the-art implementation in
CGAL~\cite{bit-hdt-22}, which was compiled with support enabled for
the CORE library for robust numeric and geometric
computation~\cite{klpy-clrngc-99}.  We note that this technique
performs all computations in the Poincar\'{e} disk model of the
hyperbolic plane (utilizing an algorithm for computing the Euclidean
Voronoi diagram followed by a post-processing step in which edges not
belonging to the hyperbolic diagram are removed), meaning all
coordinates of the sampled sites had to be converted from the
polar-coordinate model to this one, and the output of the computations
(the Voronoi vertices) had to be converted back.  This, of course,
leads to a disadvantage as the conversions involve an exponential
decrease and increase of the values, respectively, which can already
be affected by numerical inaccuracies.  We refer to this method as
\emph{CGAL (Converted)}.  Apart from that, we considered our algorithm
performing all computations natively in the polar-coordinate model of
the hyperbolic plane, and distinguished between different number
representations.  The first variant uses the \emph{IEEE 754
  Double-precision floating-point format}, which we refer to as
\emph{Native (Double)}.  Additionally, we considered variants where
all computations are performed using the \emph{multiple-precision
  floating-point library MPFR}~\cite{fhl-m-07}.  There, we set the
precisions to 64 bits and 128 bits, and denote the corresponding
techniques with \emph{Native (64)} and \emph{Native (128)},
respectively.  We note that all values are converted to Double
precision at the end, to get comparable outputs.

\begin{figure}[t!]
  \centering
  \input{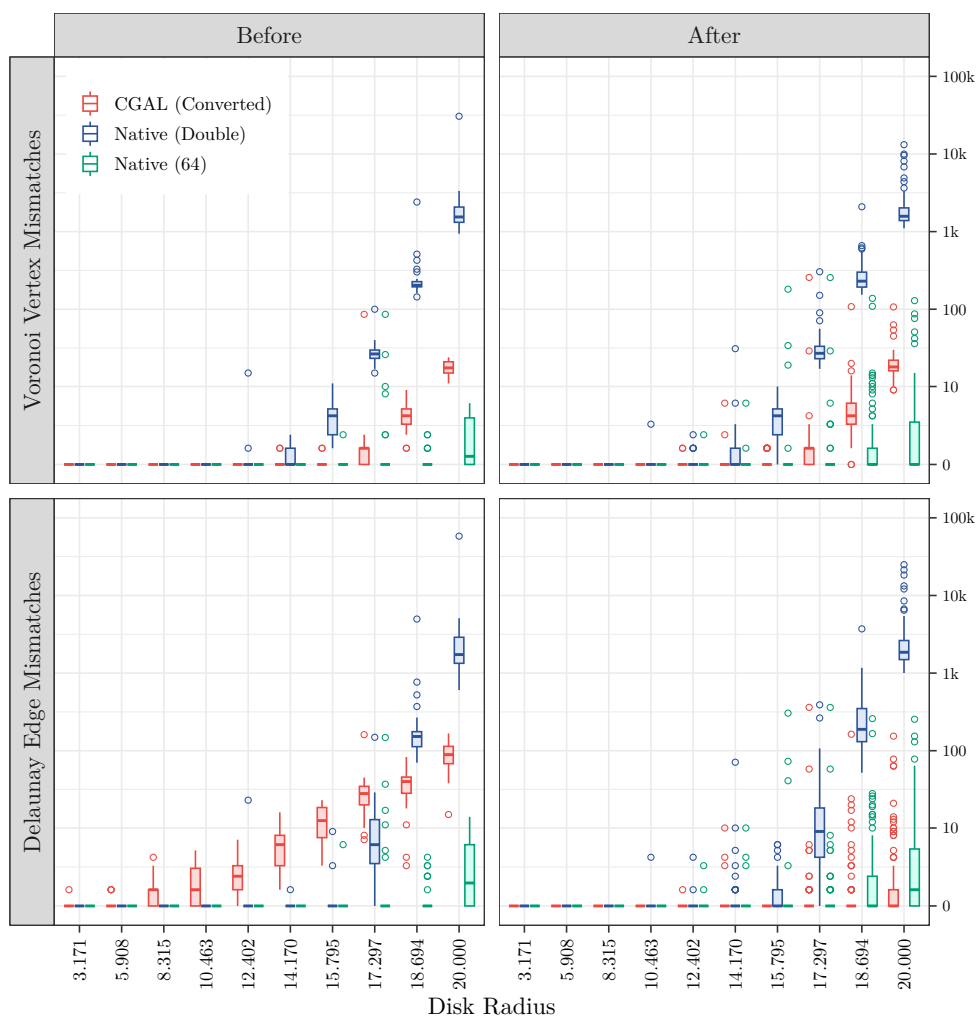}
  \caption{Comparing the robustness of different techniques for
    computing hyperbolic Voronoi diagrams and Delaunay complexes on
    sites that are distributed uniformly in disks of increasing radii.
    \textbf{(Top)} The number of Voronoi vertices present in the
    Native (128) diagram but \emph{not} in the one obtained using the
    considered technique. \textbf{(Bottom)} The number of Delaunay
    edges present in the Native (128) complex but \emph{not} in the
    one obtained by the considered the technique. \textbf{(Left)}
    Values obtained before we reported the issue. \textbf{(Right)}
    Values obtained after the issue was fixed.}
  \label{fig:experiments-comparison}
\end{figure}

Unfortunately, there is no ground truth that we can compare these
methods to.  Instead, we consider the Native (128) variant of our
algorithm as the one with the highest precision and compared all other
techniques to this one.  The differences between the outputs were then
quantified with respect to two measures.  The first denotes for how
many Voronoi vertices in the diagram computed using the Native (128)
method there was no exact match in the diagram computed using the
other technique, which we refer to as \emph{Voronoi vertex
  mismatches}.  Since comparing the structure of two diagrams is not
as straight-forward, we utilize the dual of the diagram, the Delaunay
complex, instead.  Thus, the second measure considers the graph
representing the Delaunay complex computed using the Native~(128)
method, and counts how many of its edges were not found in the complex
computed using the other technique, which we refer to as
\emph{Delaunay edge mismatches}.

Again, since there is no ground truth, we cannot know if either
implementation is correct.  However, Native (128) and CGAL (Converted)
produced the exact same Voronoi vertices for all $\num{800}$ instances
with a disk radius of at most $\num{10.463}$, which we refer to as the
\emph{small disks}.  Assuming that computations on larger disks were
affected by numerical inaccuracies, both implementations appear to be
equally correct in this regard.  For the Delaunay complexes the
situation is different.  On all small disks CGAL (Converted) yielded
the same graph \emph{or a subgraph} of the one obtained using Native
(128).  However, some of the graphs computed using CGAL (Converted)
were not connected (although, theoretically, the Delaunay complex
always is), from which we were able to infer that this is not an issue
with the Native (128) solution.  As mentioned before, we reported the
issue and a fix was supplied afterwards.

Figure~\ref{fig:experiments-comparison} shows the box plots
summarizing our experiments.  There, the considered disk radii are
shown on the $x$-axis, with the number of mismatches on the $y$-axis
in a logarithmic scale.  For each disk radius we depict three box
plots (distinguished by colors) aggregating the $\num{100}$ values
(one for each of the instances with that radius) obtained using the
different techniques.  Boxes extend to the $25$th and $75$th
percentile with the median shown as a horizontal bar, while whiskers
extend to $3/2$ times the interquartile range above and below the
boxes.  Circles denote values outside of this range.

We consider the Voronoi vertex mismatches first.  As shown in
Figure~\ref{fig:experiments-comparison} (top row) for the small disks
(of radius at most $\num{10.463}$), the red boxes representing the
CGAL (Converted) values degenerate into horizontal bars at $0$,
supporting the fact that the computed diagrams match the ones obtained
using Native (128), exactly.  As the disk radii increase beyond that,
so does the number of considered sites in an instance and with them
the Voronoi vertex mismatches.  A similar trend can be seen for the
other two techniques as well.  Compared to CGAL (Converted) the Native
(64) method shows less mismatches, while Native (Double) yielded more
mismatches and did so on smaller disks, as witnessed by the outlier
circles for disk radii of $\num{10.463}$ and larger.  We can conclude
that the computation of the Voronoi vertices in the CGAL (Converted)
method is more robust than the Native (Double) one, despite the
conversions of the coordinates between the different models of the
hyperbolic plane.

Initially, our analysis of the Delaunay edge mismatches showed a
rather different behavior, see Figure~\ref{fig:experiments-comparison}
(bottom left).  For none of the considered disk radii did CGAL
(Converted) reliably match the edges obtained using Native~(128).
Even for the smallest considered disk radius of $\num{3.171}$ there
were instances where an edge was not found, as illustrated in
Figure~\ref{fig:computed-drawings}~(right).  As with the Voronoi
vertices, the edge mismatches then increase with increasing disk
radii.  As mentioned above, this eventually led to disconnected
complexes, which is not possible in theory.  For the Native (Double)
method, note that the Delaunay edge mismatches align with the
corresponding Voronoi vertex mismatches.  That is, while the edge
mismatches increase with increasing radii, we see no mismatches for
the disk radii up to $\num{8.315}$.  The Native (64) method behaves
similarly, although mismatches only start to occur at larger disk
radii.

The fact that CGAL (Converted) did report the correct Voronoi vertices
for smaller disk radii but not the correct Delaunay edges hinted at a
possible issue in the implementation, which was reported and
subsequently fixed.  Consequently, we performed all experiments again
with the corrected implementation.  The results are shown in
Figure~\ref{fig:experiments-comparison}~(right).  As can be seen, the
number of edge mismatches in CGAL (Converted) are reduced drastically
and are now aligned with the corresponding Voronoi vertex mismatches.
This indicates that the mismatches are a result of numerical
inaccuracies.

\section{Conclusion \& Outlook}
\label{sec:outlook}

We present the first algorithm for computing Voronoi diagrams natively
in the polar-coordinate model of the two-dimensional hyperbolic plane.
We note that the distance function in this model generalizes nicely to
higher dimensions, as only the computation of the angular distance
between two points changes (one has to compute the central angle).
Consequently, we believe that the hyperbolic sweep circle approach can
be extended to higher dimensions.  Furthermore, future work may
consider extending the method to allow for computations of
higher-order Voronoi diagrams, as they are useful in the context of
nearest-neighbor queries.

Independently, the implementation of our algorithm turned out to be of
interest in general as a reference that can be used to evaluate the
correctness of existing implementations, highlighting that ease of
implementation can be an important criterion when assessing the
practicability of an algorithm.  On the other hand, this emphasizes
how much effort goes into maintaining fast and reliable
state-of-the-art implementations, for which we are grateful.

Unfortunately, our experimental evaluation shows that neither the
current state-of-the-art implementation for computing hyperbolic
Voronoi diagrams and Delaunay complexes nor our solution can be used
to reliably solve the problem in the polar-coordinate model of the
hyperbolic plane, if the considered disk radii are sufficiently large.
In particular, they are not suitable for typical applications in
network science, like the aforementioned hyperbolic random graphs,
where networks of 25k nodes already require disk sizes that exceed
what we considered in our experiments.

Thus, since none of the considered approaches seems to be able to
scale to such large areas, further extensions are necessary to make
them more robust.  Alternatively, it may prove worthwhile to examine
adaptations of other approaches to computing Voronoi diagrams, like
randomized incremental construction~\cite{gks-rv-90}, divide and
conquer~\cite{sh-c-75}, utilizing variants of abstract Voronoi
diagrams~\cite{k-cavd-93}, or a topology-oriented
method~\cite{si-rtoiavd-94}.

Nevertheless, our method is the first stepping stone towards applying
hyperbolic Voronoi diagrams and Delaunay complexes in the ongoing
study of complex networks utilizing the polar-coordinate model of the
hyperbolic plane.  In the following, we briefly highlight a direct
application of our algorithm in this context, which is the computation
of the hyperbolic counterpart to Euclidean minimum spanning trees.  To
the best of our knowledge \emph{hyperbolic random minimum spanning
  trees} have not been studied before, but may prove useful as the
tree equivalent to the aforementioned hyperbolic random graph
model~\cite{kpf-hgcn-10}.  Analogous to the Euclidean version, the
hyperbolic minimum spanning tree is a subgraph of the Delaunay complex
of a given set of sites, meaning it can be computed in time
$\mathcal{O}(n \log(n))$ using our approach.  Consequently, it would
be interesting to utilize hyperbolic Delaunay complexes in the
polar-coordinate model to generate and investigate trees in the
hyperbolic plane.



\bibliography{lipics-v2021-sample-article}

\newpage

\appendix

\section{Missing Proofs}
\label{apx:missing-proofs}

In the following, we give the proofs that were left out of the main
part of the paper due to space constraints.

\subsection{Preliminaries}

\begin{backInTime}{lem-angle-affects-distance}
  \begin{lemma}
    Let $\Triangle{A}{B}{C}$ and $\Triangle{A}{B}{C'}$ be triangles
    with $\Dist{B}{C} = \Dist{B}{C'}$ and angles $\varphi$
    and~$\varphi'$ at $B$, respectively.  Then, $\varphi < \varphi'$
    (resp. $\varphi > \varphi'$) if and only if
    $\Dist{A}{C} < \Dist{A}{C'}$ (resp. $\Dist{A}{C} > \Dist{A}{C'}$).
  \end{lemma}
  \begin{proof}
    We prove the claim for the case where $\varphi < \varphi'$.  The
    proof for the other case is analogous.  We start by showing that
    changing the angle also changes the length of the line segment
    accordingly.  Since $\varphi, \varphi'$ are inner angles of
    triangles, we have $\varphi, \varphi' \in [0, \pi)$ and thus
    $\cos(\varphi) > \cos(\varphi')$.  We can now determine
    $\Dist{A}{C}$ using the hyperbolic law of cosines and make use of
    the fact that $\Dist{B}{C} = \Dist{B}{C'}$, which yields
    \begin{align*}
      \Dist{A}{C} &= \acosh \left( \cosh(\Dist{A}{B})\cosh(\Dist{B}{C}) - \sinh(\Dist{A}{B})\sinh(\Dist{B}{C}) \cos(\varphi) \right) \\
                  &= \acosh \left( \cosh(\Dist{A}{B})\cosh(\Dist{B}{C'}) - \sinh(\Dist{A}{B})\sinh(\Dist{B}{C'}) \cos(\varphi) \right) \\
                  &< \acosh \left( \cosh(\Dist{A}{B})\cosh(\Dist{B}{C'}) - \sinh(\Dist{A}{B})\sinh(\Dist{B}{C'}) \cos(\varphi') \right) \\
                  &= \Dist{A}{C'},
    \end{align*}
    where the inequality is due to the fact that $\cos(\varphi) >
    \cos(\varphi')$ and that $\acosh(x)$ is strictly increasing with
    increasing $x$.
    
    It remains to prove that changing the length of the line segment
    also changes the angle accordingly.  Note that $\Dist{A}{C} <
    \Dist{A}{C'}$ implies $\cosh(\Dist{A}{C}) < \cosh(\Dist{A}{C'})$,
    since $\cosh(x)$ is strictly increasing for increasing $x \ge 0$.
    Again, using the hyperbolic law of cosines, we can express $\varphi$
    using the lengths of the line segments in $\Triangle{A}{B}{C}$ as
    \begin{align*}
      \varphi &= \acos \left( \frac{\cosh(\Dist{A}{B})\cosh(\Dist{B}{C}) - \cosh{\Dist{A}{C}}}{\sinh(\Dist{A}{B}) \sinh(\Dist{B}{C})} \right) \\
              &= \acos \left( \frac{\cosh(\Dist{A}{B})\cosh(\Dist{B}{C'}) - \cosh{\Dist{A}{C}}}{\sinh(\Dist{A}{B}) \sinh(\Dist{B}{C'})} \right) \\
              &< \acos \left( \frac{\cosh(\Dist{A}{B})\cosh(\Dist{B}{C'}) - \cosh{\Dist{A}{C'}}}{\sinh(\Dist{A}{B}) \sinh(\Dist{B}{C'})} \right) \\
              &= \varphi',
    \end{align*}
    where the inequality is due to the fact that $\cosh(\Dist{A}{C}) <
    \cosh(\Dist{A}{C'})$ and that $\acos(x)$ is strictly decreasing for
    increasing $x$.
  \end{proof}
\end{backInTime}

\begin{backInTime}{lem-distance-continuous}
  \begin{lemma}
    Let $\Triangle{A}{B}{C}$ be a triangle with $\Dist{A}{B} \le
    \Dist{A}{C}$.  For every $d \in [\Dist{A}{B}, \Dist{A}{C}]$, there
    exists a point $P \in \LineSegment{B}{C}$ with $\Dist{A}{P} = d$.
  \end{lemma}
  \begin{proof}
    For $d = \Dist{A}{B}$ or $d = \Dist{A}{C}$, the points $P = B$ and
    $P = C$ fulfill the requirements of the lemma, respectively.  So
    assume that $d \in (\Dist{A}{B}, \Dist{A}{C})$.  Now consider the
    circle $\Circle{A}$ of radius $d$ around $A$.  Since $\Dist{A}{B} <
    d$, the point $B$ is inside the circle.  Moreover, since
    $\Dist{A}{C} > d$, the point $C$ is outside the circle.
    Consequently, the line segment $\LineSegment{B}{C}$ intersects the
    circle at a point $P =
    \Intersection{\LineSegment{B}{C}}{\Circle{A}}$.  In particular, $P
    \in \Circle{A}$ and thus $\Dist{A}{P} = d$.
  \end{proof}
\end{backInTime}

\subsection{Beach Curve}

\begin{backInTime}{lem-ellipse-curve}
  \begin{lemma}
    Let $\SweepRadius > 0$ be the radius of the sweep circle, let $S \in
    \SitesInSweep$ be a site with $r(S) < \SweepRadius$, and let $P \in
    \Hplane$ be a point with $\varphi(P) = \varphi$.  Then, $P \in
    \Ellipse{S}$ if and only if $P$ has radius
    \begin{align*}
      \EllipseFunction{S}{\varphi} = \atanh \left(\frac{\cosh(\SweepRadius) - \cosh(r(S))}{\sinh(\SweepRadius) - \sinh(r(S))\cos(\varphi - \varphi(S))}\right ).
    \end{align*}
  \end{lemma}
  \begin{proof}
    By definition, we have $P \in \Ellipse{S}$ if and only if $P$ is
    equidistant to $S$ and the sweep circle.  Since the distance between
    a point and a circle centered at the origin is given by the
    difference of their radii, $P$ needs to fulfill the equality
    \begin{align*}
      \Dist{P}{S} = \SweepRadius - r(P).
    \end{align*}
    Applying Equation~\ref{eq:hyperbolic-distance}, which describes the
    hyperbolic distance between two points, together with an application
    of the hyperbolic cosine on both sides then yields
    \begin{align*}
      \cosh(r(P))\cosh(r(S)) - \sinh(r(P))\sinh(r(S))\cos(\varphi(P) - \varphi(S)) = \cosh(\SweepRadius - r(P)).
    \end{align*}
    We can now apply the identity $\cosh(x - y) = \cosh(x)\cosh(y) -
    \sinh(x)\sinh(y)$ to the right hand side and obtain
    \begin{multline*}
      \cosh(r(P))\cosh(r(S)) - \sinh(r(P))\sinh(r(S))\cos(\varphi(P) - \varphi(S)) =\\
      \cosh(\SweepRadius)\cosh(r(P)) - \sinh(\SweepRadius)\sinh(r(P)).
    \end{multline*}
    We continue by subtracting $\cosh(\SweepRadius)\cosh(r(P))$ on both
    sides and adding $\sinh(\SweepRadius)\sinh(r(P))$, which yields
    \begin{multline*}
      \sinh(\SweepRadius)\sinh(r(P)) - \sinh(r(P))\sinh(r(S))\cos(\varphi(P) - \varphi(S)) =\\
      \cosh(\SweepRadius)\cosh(r(P)) - \cosh(r(P))\cosh(r(S)).
    \end{multline*}
    By factoring out $\sinh(r(P))$ and $\cosh(r(P))$, we get
    \begin{align*}
      \sinh(r(P)) \big( \sinh(\SweepRadius) - \sinh(r(S))\cos(\varphi(P) - \varphi(S)) \big) = \cosh(r(P)) \big(\cosh(\SweepRadius) - \cosh(r(S)) \big),
    \end{align*}
    which is equivalent to 
    \begin{align*}
      \frac{\sinh(r(P))}{\cosh(r(P))} = \frac{\cosh(\SweepRadius) - \cosh(r(S))}{\sinh(\SweepRadius) - \sinh(r(S))\cos(\varphi(P) - \varphi(S))},
    \end{align*}
    after dividing by $\cosh(r(P))$ and $\sinh(\SweepRadius) -
    \sinh(r(S))\cos(\Delta_\varphi(P, S))$ on both sides.  Finally, we
    recognize that the left hand side is the hyperbolic tangent
    $\tanh(r(P))$.  Applying the inverse hyperbolic tangent then yields
    the claim.
  \end{proof}
\end{backInTime}

\subsection{Correctness and Complexity}
\label{sec:algorithm-correctness}

In this section, we present the components for the proof of our main
theorem.

\begin{backInTime}{thm-main}
\begin{theorem}
  Let $\Set{S} = \{S_1, \dots, S_n\} \subset \Hplane$ be a set of
  sites.  Then, the sweep circle algorithm computes $\Vor{\Set{S}}$ in
  time $\mathcal{O}(n \log(n))$.
\end{theorem}
\end{backInTime}

Like the algorithm itself, the proof of its correctness works
analogous to the Euclidean version~\cite[Section 7.2]{bcko-cg-08}.
Therefore, we only focus on the main parts of the proof and show that
they also hold for the hyperbolic sweep circle approach.  In
particular, we show that the only way in which a new beach ellipse
segment can appear on the beach curve is through a site event
(Lemma~\ref{lem:segments-appear-through-sites}), the only way in which
an existing segment can disappear from the beach curve is through a
circle event (Lemma~\ref{lem:segments-disappear-through-circles}), and
that every Voronoi vertex is detected by means of a circle event
(Lemma~\ref{lem:vertices-through-circle-events}).

\subsubsection{Beach Ellipse Intersections}
\label{sec:intersections}


To start, we establish some basic properties of beach ellipse
intersections in the hyperbolic plane, beginning with their existence.
Given a site $S \in \SitesInSweep$, we say that a point $P$ is
\emph{inside} the beach ellipses $\Ellipse{S}$, if $\Dist{P}{S}$ is
smaller than the distance between $P$ and the sweep circle.  If
$\Dist{P}{S}$ is larger instead, we say that $P$ is \emph{outside} of
$\Ellipse{S}$.  The following lemma shows that, given two distinct
sites inside the sweep circle, the beach ellipse of one contains two
points that are inside and outside of the ellipse of the other,
respectively, as illustrated in Figure~\ref{fig:intersections} (left).

\begin{figure}
  \centering
  \includegraphics{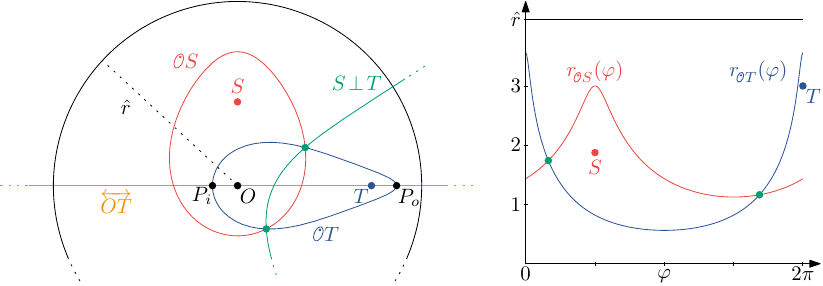}
  \caption{\textbf{(Left)} Two sites $S, T$ inside the sweep circle
    (black).  The points $P_i$ and $P_o$ are on $\Ellipse{T}$ (blue)
    and are inside and outside of $\Ellipse{S}$ (red), respectively.
    The intersections (green points) of the ellipses are on the
    perpendicular bisector $\Bisector{S}{T}$ (green) and on opposite
    sites of the line $\Line{O}{T}$ (orange). \textbf{(Right)} The
    functions $\EllipseFunction{S}{\varphi}$ and
    $\EllipseFunction{T}{\varphi}$ describe the arcs of the beach
    ellipses on the left.}
  \label{fig:intersections}
\end{figure}

\begin{lemma}
  \label{lem:inside-outside}
  Let $\SweepRadius$ be the radius of the sweep circle, let $S \neq T
  \in \SitesInSweep$ be two sites with $0 < r(S) \le r(T) <
  \SweepRadius$ and let $\Ellipse{S}$ and $\Ellipse{T}$ be their beach
  ellipses.  Then the point $P_i \in \Ellipse{T}$ with $\varphi(P_i) =
  \varphi(T) + \pi$ is inside of $\Ellipse{S}$ and $P_o \in
  \Ellipse{T}$ with $\varphi(P_o) = \varphi(T)$ is outside of
  $\Ellipse{S}$.
\end{lemma}
\begin{proof}
  Assume, without loss of generality, that $\varphi(T) = 0$.  To show
  that $P_i$ is inside of $\Ellipse{S}$, we show that the distance
  between $P_i$ and the sweep circle is larger than $\Dist{P_i}{S}$.
  More precisely, since $P_i \in \Ellipse{T}$, $P_i$ is equidistant to
  the sweep circle and to $T$.  Consequently, it suffices to show that
  $\Dist{P_i}{T} > \Dist{P_i}{S}$.  We now distinguish several cases
  depending on the position of $S$.  If $\varphi(S) = \varphi(T)$,
  then we have $r(S) < r(T)$ (since $S \neq T$) and thus
  \begin{align*}
    \Dist{P_i}{T} = r(P_i) + r(T) > r(P_i) + r(S) = \Dist{P_i}{S}.
  \end{align*}
  If $\varphi(S) = \varphi(T) + \pi$, we first consider the case where
  $r(S) \ge r(P_i)$.  Then,
  \begin{align*}
    \Dist{P_i}{T} = r(P_i) + r(T) \ge r(P_i) + r(S) > r(S) - r(P_i) = \Dist{P_i}{S}.
  \end{align*}
  Alternatively, if $\varphi(S) = \varphi(T) + \pi$ and
  $r(S) < r(P_i)$, we have
  \begin{align*}
    \Dist{P_i}{T} = \Dist{P_i}{S} + r(S) + r(T) > \Dist{P_i}{S}.
  \end{align*}
  In all other cases, we can consider the triangle
  $\Triangle{O}{S}{P_i}$.  Since the area of this triangle is
  non-zero, we can apply the strict triangle inequality, which yields
  $\Dist{P_i}{S} < r(P_i) + r(S)$.  It follows that
  \begin{align*}
    \Dist{P_i}{T} = r(P_i) + r(T) \ge r(P_i) + r(S) > \Dist{P_i}{S}.
  \end{align*}
  Consequently, in all cases $P_i$ is closer to $S$ than to the sweep
  circle, and is therefore inside~$\Ellipse{S}$.  Proving that $P_o$
  is outside of $\Ellipse{S}$ works analogously.  We show that $P_o$
  is closer to the sweep circle than to $S$, which is equivalent to
  showing $\Dist{P_o}{T} < \Dist{P_o}{S}$, since $P_o$ is equidistant
  to $T$ and the sweep circle.  We again distinguish several cases
  depending on the position of $S$.  If $\varphi(S) = \varphi(T)$, the
  distinctness of $S$ and $T$ implies $r(S) < r(T)$ and thus
  \begin{align*}
    \Dist{P_o}{T} = r(P_o) - r(T) < r(P_o) - r(S) = \Dist{P_o}{S}.
  \end{align*}
  Similarly, if $\varphi(S) = \varphi(T) + \pi$, we have
  \begin{align*}
    \Dist{P_o}{T} = r(P_o) - r(T) < r(P_o) + r(S) = \Dist{P_o}{S}.
  \end{align*}
  Again, in all other cases, we can apply the strict triangle
  inequality to conclude $r(P_o) < r(S) + \Dist{P_o}{S}$ or
  equivalently $r(P_o) - r(S) < \Dist{P_o}{S}$.  Thus,
  \begin{align*}
    \Dist{P_o}{T} = r(P_o) - r(T)  \le r(P_o) - r(S) < \Dist{P_o}{S},
  \end{align*}
  which shows that $P_o$ is outside of $\Ellipse{S}$ in all cases.
\end{proof}

With the above lemma we can now prove that two non-degenerate beach
ellipses intersect in exactly two points.  If one of them is
degenerate, they intersect in exactly one point.

\begin{lemma}
  \label{lem:ellipse-intersections}
  Let $\SweepRadius$ be the radius of the sweep circle, let $S \neq T
  \in \SitesInSweep$ be two sites with $0 < r(S) \le r(T)$ and let
  $\Ellipse{S}$ and $\Ellipse{T}$ be their beach ellipses.  Then,
  $\Intersection{\Ellipse{S}}{\Ellipse{T}}$ contains two points if
  $r(S), r(T) < \SweepRadius$ and one point, otherwise.
\end{lemma}
\begin{proof}
  We start with the non-degenerate case and show that the number of
  intersections is at least two and that it is at most two.  By
  Lemma~\ref{lem:ellipse-curve}, the beach ellipse of $S$ is described
  by a function $\EllipseFunction{S}{\varphi}$ that maps an angle
  $\varphi$ to the radius of the point $P \in \Ellipse{S}$ with
  $\varphi(P) = \varphi$, see Figure~\ref{fig:intersections} (right).
  Analogously, there is a function $\EllipseFunction{T}{\varphi}$ for
  $\Ellipse{T}$.  As shown in Lemma~\ref{lem:inside-outside}, there
  are two points $P_i, P_o \in \Ellipse{T}$ that are inside and
  outside of $\Ellipse{S}$, respectively.  Thus,
  $\EllipseFunction{S}{\varphi(P_i)} >
  \EllipseFunction{T}{\varphi(P_i)}$ and
  $\EllipseFunction{S}{\varphi(P_o)} <
  \EllipseFunction{T}{\varphi(P_o)}$.  Since
  $\EllipseFunction{S}{\varphi}, \EllipseFunction{T}{\varphi}$ are
  continuous and periodic functions with a period of $2\pi$ (see
  Lemma~\ref{lem:ellipse-curve}), we can apply the intermediate value
  theorem to conclude that there are at least two values
  $\varphi \in [0, 2\pi)$ such that
  $\EllipseFunction{S}{\varphi} = \EllipseFunction{T}{\varphi}$.
  
  We now show that there are at most two intersections.  The angular
  coordinates of all intersections are obtained by solving
  $\EllipseFunction{S}{\varphi} = \EllipseFunction{T}{\varphi}$ for
  $\varphi$.  By Lemma~\ref{lem:ellipse-curve}, this happens when
  \begin{align*}
    \frac{\cosh(\SweepRadius) - \cosh(r(S))}{\sinh(\SweepRadius) - \sinh(r(S))\cos(\varphi - \varphi(S))} = \frac{\cosh(\SweepRadius) - \cosh(r(T))}{\sinh(\SweepRadius) - \sinh(r(T))\cos(\varphi - \varphi(T))}.
  \end{align*}
  Note that solving this equation for $\varphi$ is equivalent to
  finding the roots of the function
    $f(\varphi) = a \cos(\varphi - \varphi(S)) - b \cos(\varphi - \varphi(T)) + c$,
  where the constants $a, b$, and $c$ are defined as
  \begin{align*}
    a &= \left( \cosh(\SweepRadius) - \cosh(r(T)) \right) \sinh(r(S)),~\text{and}~b = \left( \cosh(\SweepRadius) - \cosh(r(S)) \right) \sinh(r(T)),~\text{and} \\
    c &= \left( \cosh(r(T)) - \cosh(r(S)) \right) \sinh(\SweepRadius).
  \end{align*}
  Note that $f(\varphi)$ is the sum of two differently phased cosine
  functions of equal frequency with amplitudes $a$ and $b$,
  respectively, together with a constant $c$.  Since the sum of two
  sinusoids of the same frequency is another sinusoid, we can use
  \cite[Equation (6)]{l-sts-11}, to conclude that
    $f(\varphi) = a' \cos(\varphi + \xi) + c$,
  where the constants $a'$ and $\xi$ are given by
  \begin{align*}
    a' &= \sqrt{ \left( a \cos(\varphi(S)) + b \cos(\varphi(T)) \right)^2 + \left( a \sin(\varphi(S)) + b \sin(\varphi(T)) \right)^2}, \\
    \xi &= \atan \left( \frac{a \sin(\varphi(S)) + b \sin(\varphi(T))}{a \cos(\varphi(S)) + b \cos(\varphi(T))} \right),
  \end{align*}
  and $a, b,$ and $c$ are defined as above.  Consequently,
  $f(\varphi)$ is a cosine function with amplitude $a'$ and period
  $2\pi$, which has at most two roots in $[0, 2\pi)$ if $a' \neq 0$.
  Therefore, it remains to show that $a'$ is non-zero.  First note
  that $a'$ only vanishes if both sums in the quadratic functions do.
  Since there exists no $x \in \mathbb{R}$ such that $\cos(x) = 0$
  \emph{and} $\sin(x) = 0$, it follows that $a'$ is non-zero as long
  as $a$ and $b$ are non-zero.  Since $\Ellipse{S}$ and $\Ellipse{T}$
  are non-degenerate, we know that $r(S), r(T) < \SweepRadius$, and
  since $\cosh(x)$ is strictly increasing for $x \ge 0$, we have
  \begin{align*}
    \cosh(\SweepRadius) - \cosh(r(S)) > 0~\text{and}~\cosh(\SweepRadius) - \cosh(r(T)) > 0.
  \end{align*}
  Moreover, since $r(S), r(T) > 0$ by assumption, it follows that
  $\sinh(r(S)), \sinh(r(T)) > 0$, which concludes the proof of the
  non-degenerate case.
  
  If $\Ellipse{S}$ is degenerate, then
  $\Ellipse{S} = \LineSegment{O}{S}$ and all intersections in
  $\Intersection{\Ellipse{S}}{\Ellipse{T}}$ have angular coordinate
  $\varphi(S)$.  Thus, by Lemma~\ref{lem:ellipse-curve}, there is only
  one point $P \in \Ellipse{T}$ satisfying $\varphi(P) = \varphi(S)$.
  Analogously, there is only one intersection when $\Ellipse{T}$ is
  degenerate but $\Ellipse{S}$ is not.  Finally, when both sites are
  degenerate, i.e., $r(S) = r(T) = \SweepRadius$, then
  $\varphi(S) \neq \varphi(T)$, since both sites are assumed to be
  distinct.  In that case, the two ellipses
  $\Ellipse{S} = \LineSegment{O}{S}$ and
  $\Ellipse{T} = \LineSegment{O}{T}$ intersect only in the pole.
\end{proof}
We continue by investigating how the intersections move as the sweep
circle expands.  To this end, we first show that the beach ellipses
expand as well.

\begin{lemma}
  \label{lem:ellipse-expansion}
  Let $S \in \SitesInSweep$ be a site and let $r_1, r_2$ be two radii
  with $r(S) \le r_1 < r_2$.  Then, $\Ellipse{S}$ at $\SweepRadius =
  r_1$ is inside of $\Ellipse{S}$ at $\SweepRadius = r_2$.
\end{lemma}
\begin{proof}
  Consider a point $P \in \Ellipse{S}$ as $\SweepRadius = r_1$.  Then,
  the distance between $P$ and the sweep circle is given by
  $r_1 - r(P)$ and is equal to $\Dist{P}{S}$.  At
  $\SweepRadius = r_2$, $\Dist{P}{S}$ remains unchanged.  However,
  then the distance between $P$ and the sweep circle increases to
  $r_2 - r(P) > r_1 - r(P) = \Dist{P}{S}$.  It follows that $P$ is
  closer to $S$ than to the sweep circle at $\SweepRadius = r_2$ and
  is therefore inside~$\Ellipse{S}$.
\end{proof}

We are now ready to show that the two beach ellipse intersections of a
pair of sites $S$ and~$T$ start at the same point and move along the
bisector $\Bisector{S}{T}$ in opposite directions as the radius
$\SweepRadius$ of the sweep circle increases.

\begin{lemma}
  \label{lem:intersections-move}
  Let $S \neq T \in \SitesInSweep$ be two sites with $r(S) \le r(T)$
  and consider the intersections
  $I, I' \in \Intersection{\Ellipse{S}}{\Ellipse{T}}$ and the point
  $P = \Intersection{\LineSegment{O}{T}}{\Bisector{S}{T}}$.  For
  $\SweepRadius = r(T)$, we have $I = I' = P$.  For
  $\SweepRadius > r(T)$, $I$ and $I'$ are on opposite sides of
  $\Line{O}{T}$.  As $\SweepRadius$ increases, so do $\Dist{P}{I}$ and
  $\Dist{P}{I'}$.
\end{lemma}
\begin{proof}
  For $\SweepRadius = r(T)$, the beach ellipse $\Ellipse{T}$ is
  degenerate and consists of the line segment $\LineSegment{O}{T}$.
  By Lemma~\ref{lem:ellipse-intersections}, the two ellipses
  $\Ellipse{S}$ and $\Ellipse{T}$ intersect in a single point.  This
  point is $P$ as the intersection is on $\Ellipse{T} =
  \LineSegment{O}{T}$ and has by definition equal distance to $S$ and
  $T$, i.e., it lies on $\Bisector{S}{T}$.  Moreover, this point is
  unique, as $\LineSegment{O}{T} \not\subset \Bisector{S}{T}$, since
  $T \in \LineSegment{O}{T}$ but $T \notin \Bisector{S}{T}$.
  
  We continue with the proof that $I$ and $I'$ are on opposite sides
  of $\Line{O}{T}$ for $\SweepRadius > r(T)$, as shown in
  Figure~\ref{fig:intersections} (left).  Without loss of generality,
  assume that $\varphi(T) = 0$ and consider the points
  $P_i, P_o \in \Ellipse{T}$ as defined in
  Lemma~\ref{lem:inside-outside}, that are inside and outside of
  $\Ellipse{S}$, respectively.  That is,
  $\EllipseFunction{T}{0} > \EllipseFunction{S}{0}$ and
  $\EllipseFunction{T}{\pi} < \EllipseFunction{S}{\pi}$.
  Consequently, one of $I$ and $I'$ has an angular coordinate in
  $(0, \pi)$ and the other in $(\pi, 2\pi)$.  Thus, they are on
  opposite sides of $\Line{O}{T}$.
  
  It remains to show that $\Dist{P}{I}$ and $\Dist{P}{I'}$ increase
  with $\SweepRadius$.  Consider two radii $r_1, r_2$ with
  $r(T) \le r_1 < r_2$.  Let $I_{r_1}$ and $I_{r_2}$ denote the
  positions that the intersection $I$ has at $\SweepRadius = r_1$ and
  $\SweepRadius = r_2$, respectively, and let $I'_{r_1}$ and
  $I'_{r_2}$ be defined analogously.  Consider $\SweepRadius = r_1$
  first.  By Lemma~\ref{lem:ellipse-intersections} the bisector
  $\Bisector{S}{T}$ intersects $\Ellipse{T}$ at exactly the points
  $I_{r_1}$ and $I'_{r_1}$, meaning the line segment
  $\LineSegment{I_{r_1}}{I'_{r_1}}$ lies in $\Ellipse{T}$.  Moreover,
  the point $P$ lies in $\Ellipse{T}$ and on $\Bisector{S}{T}$,
  meaning $P \in \LineSegment{I_{r_1}}{I'_{r_1}}$ and thus the line
  segments $\LineSegment{P}{I_{r_1}}$ and $\LineSegment{P}{I'_{r_1}}$
  are contained in $\Ellipse{T}$.  Now assume that $\SweepRadius$
  increases to $\SweepRadius = r_2$ but $\Dist{P}{I}$ does not.  (The
  proof for $\Dist{P}{I'}$ is analogous.)  Then
  $I_{r_2} \in \LineSegment{P}{I_{r_1}}$.  However, since the beach
  ellipse $\Ellipse{T}$ at $\SweepRadius = r_1$ is completely
  contained in $\Ellipse{T}$ at $\SweepRadius = r_2$
  (Lemma~\ref{lem:ellipse-expansion}), it follows that
  $I_{r_2} \in \LineSegment{P}{I_{r_1}}$ is inside of $\Ellipse{T}$,
  contradicting the fact that $I_{r_2} \in \Ellipse{T}$.
\end{proof}

The above lemma has some interesting implications.  Consider two sites
$S, T$ that are incident to a Voronoi vertex $V$, meaning they lie on
the arc of the witness circle of $V$.  Then,~$V$ lies on the bisector
$\Bisector{S}{T}$ and, in particular, on one side of the point
$P \in \Bisector{S}{T}$ defined in the lemma.  As the sweep circle
expands, the intersections $I$ and $I'$ move along the bisector and
away from $P$, meaning exactly one of them reaches $V$ eventually.
More precisely, the following lemma captures from which directions the
intersections approach a Voronoi vertex.  Recall that
$\Angle{A}{B}{C}$ denotes the angle between $\Ray{B}{A}$ and
$\Ray{B}{C}$ in clockwise direction around $B$.

\begin{lemma}
  \label{lem:intersection-direction}
  Let $V \in \Set{V}$ be a Voronoi vertex with far point $V'$ and
  incidence tuple $(S_1, S_2, S_3)$.  Further, let $(S, S') \in
  \{(S_1, S_2), (S_2, S_3), (S_1, S_3)\}$, let $\SingleRay{V}$ be the
  angular bisector of $\Angle{S}{V}{S'}$, and let $I \in
  \Intersection{\Ellipse{S}}{\Ellipse{S'}}$ be the intersection with
  $I = V$ at $\SweepRadius = r(V')$.  Then, $I \in \SingleRay{V}$ for
  $\SweepRadius < r(V')$.
\end{lemma}
\begin{proof}
  \begin{figure}[t]
    \centering
    \includegraphics{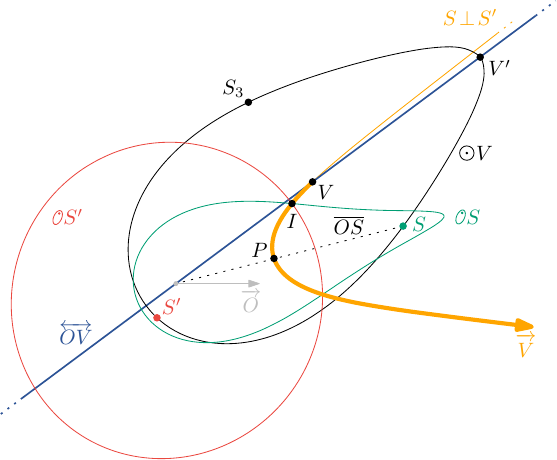}
    \caption{Illustration of the proof of
      Lemma~\ref{lem:intersection-direction}.  The sites $S$ and $S'$
      lie on the witness circle of~$V$.  The intersection
      $I \in \Intersection{\Ellipse{S}}{\Ellipse{S'}}$ that moves
      towards $V$ lies on the angular bisector $\SingleRay{V}$
      (orange) of the angle $\Angle{S}{V}{S'}$, since $\SingleRay{V}$,
      $I, P$, and $S$ are on the same side of $\Line{O}{V}$.}
    \label{fig:directions}
  \end{figure}

  Note that $V$ divides $\Bisector{S}{S'}$ into the two rays
  $\SingleRay{V}$ and $\Bisector{S}{S'} \setminus \SingleRay{V}$ and
  so does the line $\Line{O}{V}$.  Thus, since
  $I \in \Bisector{S}{S'}$, it suffices to show that $I$ and
  $\SingleRay{V}$ are on the same side of $\Line{O}{V}$.  Without loss
  of generality, assume $r(S) \ge r(S')$.  By
  Lemma~\ref{lem:intersections-move} there is a point
  $P = \Intersection{\LineSegment{O}{S}}{\Bisector{S}{S'}}$ such that
  $I \in \LineSegment{P}{V} \subset \Bisector{S}{S'}$.  Thus, $P$ and
  $I$ are on the same side of~$\Line{O}{V}$.  Moreover, since
  $P \in \LineSegment{O}{S}$, we also know that $S$ and $I$ are on the
  same side of~$\Line{O}{V}$.  Consequently, it suffices to show that
  $\SingleRay{V}$ is on the same side of $\Line{O}{V}$ as $S$, as
  shown in Figure~\ref{fig:directions}.

  If $S'$ is on the same side of $\Line{O}{V}$ as $S$, then this is
  trivially true, since $\SingleRay{V}$ is the angular bisector of
  $\Angle{S}{V}{S'}$.  Thus, let $S$ and $S'$ lie on opposite sides of
  $\Line{O}{V}$ and consider the angles $\varphi_S = \Angle{S}{V}{O}$
  and $\varphi_{S'} = \Angle{O}{V}{S'}$.  Since $\SingleRay{V}$ is the
  angle bisector of $\Angle{S}{V}{S'}$, we know that the angle between
  $\Ray{V}{S}$ and $\SingleRay{V}$ is
  $1/2 (\varphi_S + \varphi_{S'})$.  To show that $\SingleRay{V}$ is
  on the same side of $\Line{O}{V}$ as $S$, it then suffices to show
  that $1/2(\varphi_S + \varphi_{S'}) \le \varphi_S$ or equivalently
  that $\varphi_{S'} \le \varphi_{S}$.  To this end, we make use of
  the hyperbolic law of cosines.

  Consider the triangles $\Triangle{O}{V}{S}$ and
  $\Triangle{O}{V}{S'}$ and note that $\varphi_S$ and $\varphi_{S'}$
  are the angles at $V$, respectively.  Moreover, since $S$ and $S'$
  lie on the witness circle of $V$, we have
  $\Dist{V}{S} = \Dist{V}{S'}$.  If $r(S') = r(S)$, the side lengths
  of the triangles match and we have $\varphi_S = \varphi_{S'}$ by
  Equation~\eqref{eq:law-of-cosines}.  If $r(S') < r(S)$, i.e., if
  $\Dist{O}{S'} < \Dist{O}{S}$, we can apply
  Lemma~\ref{lem:angle-affects-distance} to conclude that
  $\varphi_{S'} < \varphi_S$.
\end{proof}

As a corollary of the above lemma, we can conclude that before the
sweep circle radius reaches the far point $V'$ of $V$, there are two
intersections whose distance to the sites of the intersecting beach
ellipses is smaller than to the third site, as shown in
Figure~\ref{fig:direction-distance}.

\begin{figure}[t]
  \centering
  \includegraphics{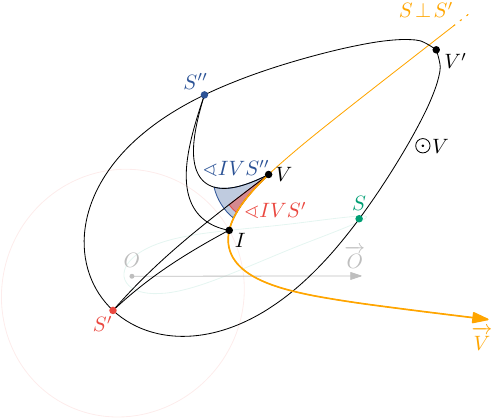}
  \caption{Illustration of Corollary~\ref{cor:third-site-distinct}.
    The distance from $I$ to $S'$ is smaller than the one to $S''$.}
  \label{fig:direction-distance}
\end{figure}

\begin{corollary}
  \label{cor:third-site-distinct}
  Let $V \in \Set{V}$ be a Voronoi vertex with far point $V'$ and
  incidence tuple $(S_1, S_2, S_3)$.  Further, let
  $(S, S', S'') \in \{(S_1, S_2, S_3), (S_2, S_3, S_1)\}$ and consider
  the intersection $I \in \Intersection{\Ellipse{S}}{\Ellipse{S'}}$
  with $I = V$ at $\SweepRadius = r(V')$.  Then,
  $\Dist{I}{S''} > \Dist{I}{S} = \Dist{I}{S'}$ for
  $\SweepRadius < r(V')$.
\end{corollary}
\begin{proof}
  Without loss of generality, assume that
  $\Angle{S'}{V}{S''} \le \Angle{S''}{V}{S}$.  Consider the two
  triangles $\Triangle{I}{V}{S'}$ and $\Triangle{I}{V}{S''}$, as
  illustrated in Figure~\ref{fig:direction-distance}, and note that
  $\Dist{V}{S'} = \Dist{V}{S''}$, since $S'$ and $S''$ lie on the
  witness circle of $V$.  By Lemma~\ref{lem:intersection-direction},
  we know that $I$ lies on the angle bisector $\SingleRay{V}$ of
  $\Angle{S}{V}{S'}$, meaning $I$ lies between $\Ray{V}{S}$ and
  $\Ray{V}{S'}$ in clockwise direction around $V$.  By definition of
  the incidence tuple, we know that $S''$ does not lie between
  $\Ray{V}{S}$ and $\Ray{V}{S'}$ in clockwise direction, meaning
  $\Angle{I}{V}{S'} < \Angle{I}{V}{S''}$.  Consequently, we can apply
  Lemma~\ref{lem:angle-affects-distance} to conclude that
  $\Dist{I}{S} = \Dist{I}{S'} < \Dist{I}{S''}$.
\end{proof}

Finally, we investigate how we can use the beach ellipse intersections
to predict circle events and how to distinguish between true and false
ones (see Section~\ref{sec:event-queue}).  Predicting the event is
straightforward, as we only need to compute the intersection $P$ of
the two bisectors corresponding to two beach ellipse intersections
that are consecutive\footnote{We argue in the proof of
  Lemma~\ref{lem:vertices-through-circle-events} that it suffices to
  consider consecutive intersections.} on the beach curve.  We note
that $P$ may not exist, in which case no circle event is predicted.
If it does exist, we need to determine whether the two beach ellipse
intersections converge towards $P$ as the sweep circle expands.  To
this end, let $S_1$ and $S_2$ be two sites and recall that an
intersection $I \in \Intersection{\Ellipse{S_1}}{\Ellipse{S_2}}$ moves
away from the line through the pole and the site with the larger
radius (Lemma~\ref{lem:intersections-move}).  We call this site the
\emph{dominant site} of $I$.  The following lemma now says that, as
the sweep circle expands, two beach ellipse intersections meet at a
point $P$ (predicting a true circle event), if $P$ and the
intersections are on the same side of the lines through the pole and
the dominant sites.  See Figure~\ref{fig:intersections-converge} for
an illustration.

\begin{lemma}
  \label{lem:true-false-circle-detection}
  Let $S_1, S_2, S_3 \in \SitesInSweep$ be distinct and let
  $P \in \Hplane$ lie at distance $d$ to them.  Further, let
  $r < r(P) + d$ be such that
  $I \in \Intersection{\Ellipse{S_1}}{\Ellipse{S_2}}$ and
  $I' \in \Intersection{\Ellipse{S_2}}{\Ellipse{S_3}}$ are distinct at
  $\SweepRadius = r$ and let $S$ and $S'$ be their dominant sites,
  respectively.  Then, $P = I = I'$ at $\SweepRadius = r(P) + d$, if
  and only if $P$ and $I$ (resp. $I'$) are on the same side of
  $\Line{O}{S}$ (resp. $\Line{O}{S'}$) at $\SweepRadius = r$.
\end{lemma}
\begin{proof}
  We give the proof for $I$ and $\Line{O}{S}$.  The one for $I'$ and
  $\Line{O}{S'}$ is analogous.  Note that the positions of the sites
  and thus the coordinates of $S$ and $P$ are fixed.  Consequently,
  $P$ is on the same side of $\Line{O}{S}$ at all times.  Moreover, by
  Lemma~\ref{lem:intersections-move} the intersection $I$ is on the
  same side of $\Line{O}{S}$ at all times.  It follows that, if and
  only if $P$ and $I$ are on the same side of $\Line{O}{S}$ at a given
  sweep circle radius, then this holds for all sweep circle radii.
  
  When $\SweepRadius = r(P) + d$, the sweep circle has equal distance
  to $P$ as to all sites $S_1, S_2$, and $S_3$, meaning $P$ lies on
  their beach ellipses.  In particular, we have $P = I$.  Clearly, $P$
  and $I$ lie on the same side of $\Line{O}{S}$ at that point.  By the
  above argumentation, then and only then does the same hold at
  $\SweepRadius = r$.
\end{proof}

\begin{figure}[t]
  \centering
  \includegraphics{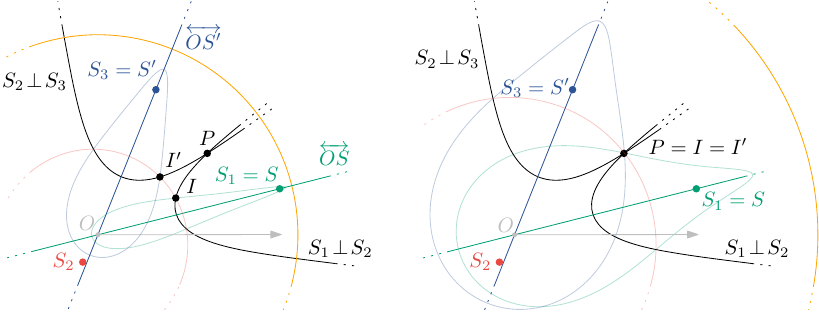}
  \caption{Illustration of
    Lemma~\ref{lem:true-false-circle-detection}.  \textbf{(Left)} The
    sweep circle (orange) has radius $\SweepRadius < r(P) + d$.  The
    point $P$ and the intersections $I$ and $I'$ are on the same sides
    of $\Line{O}{S}$ and $\Line{O}{S'}$, respectively.
    \textbf{(Right)} The intersections meet at $P$ when
    $\SweepRadius = r(P) + d$.}
  \label{fig:intersections-converge}
\end{figure}

\subsubsection{Active Beach Ellipse Segments}

In this section, we consider how the beach curve changes as the sweep
circle expands.  We start by proving the following lemma, which
characterizes when beach ellipse intersections are active, i.e., when
they are on the beach curve $\Set{B}$, as depicted in
Figure~\ref{fig:active-segments}.

\begin{figure}[t]
  \centering
  \includegraphics{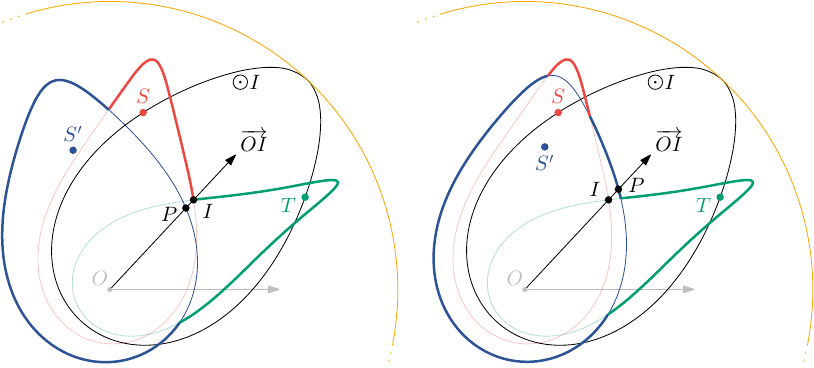}
  \caption{Illustration of the proof of
    Lemma~\ref{lem:active-means-closer}.  The circle $\Circle{I}$
    contains all points that lie at equal distance to $I$ as $S$ and
    $T$.  Note that this circle is tangent to the sweep circle
    (orange).  \textbf{(Left)} The intersection $I$ is on the beach
    curve (bold), since $S$ and $T$ are closer to $I$ than $S'$.
    \textbf{(Right)} The intersection $I$ is not on the beach curve,
    since $S'$ is closer to~$I$ than $S$ and $T$}
  \label{fig:active-segments}
\end{figure}

\begin{lemma}
  \label{lem:active-means-closer}
  Let $S, T \in \SitesInSweep$ be two sites and let
  $I \in \Intersection{\Ellipse{S}}{\Ellipse{T}}$ be an intersection
  of their ellipses.  Then, $I \in \Set{B}$ if and only if there
  exists no site $S' \in \SitesInSweep$ with $S' \neq S, T$ such that
  $\Dist{S'}{I} < \Dist{S}{I} = \Dist{T}{I}$.
\end{lemma}
\begin{proof}
  We start by proving that $S'$ does not exist if $I \in \Set{B}$.  To
  this end, we show for each $S' \in \SitesInSweep$ with
  $S' \neq S, T$, that $\Dist{S'}{I} \ge \Dist{S}{I} = \Dist{T}{I}$.
  Consider the intersection
  $P \in \Intersection{\Ray{O}{I}}{\Ellipse{S'}}$.  Since
  $I \in \Set{B}$, we know that $r(I) \ge r(P)$, meaning $I$ is not
  inside the beach ellipse $\Ellipse{S'}$.  Thus, $I$ is at least as
  close to the sweep circle as to $S'$.  Since the distance between
  $I$ and the sweep circle is given by $\SweepRadius - r(I)$, it
  follows that
  $\Dist{S'}{I} \ge \SweepRadius - r(I) = \Dist{S}{I} = \Dist{T}{I}$.
  
  It remains to consider the case where $I \notin \Set{B}$.  Then,
  there exists a beach ellipse segment of another site $S' \neq S, T$
  that is active at angular coordinate $\varphi(I)$.  That is, there
  is a point $P = \Intersection{\Ray{O}{I}}{\Ellipse{S'}}$ such that
  $r(I) < r(P)$.  It follows that $I$ is inside of $\Ellipse{S'}$,
  meaning $I$ is closer to $S'$ than to the sweep circle.  We can
  conclude that
  $\Dist{S'}{I} < \SweepRadius - r(I) = \Dist{S}{I} = \Dist{T}{I}$.
\end{proof}

With the above lemma, we are now ready to investigate how changes to
the beach curve are related to the events in the queue $\Set{Q}$.
Consider the tuple $\Set{\bar{Q}}$ containing all site events and all
true circle events (see Section~\ref{sec:event-queue}).  Recall that
$V'$ denotes the far point of a Voronoi vertex $V$, i.e., the point
with the maximum radial coordinate among points on the witness circle
of $V$.  Then, $\Set{\bar{Q}} = (r_0, r_1, \dots, r_k, r_{k + 1})$
contains the radii $r_0 = 0$, $r_{k + 1} = \infty$ and all radii
$\{ r(S) \mid S \in \Set{S} \} \cup \{ r(V') \mid V \in \Set{V}\}$ in
ascending order inbetween.  We note that $\Set{\bar{Q}}$ is different
from $\Set{Q}$, since the latter also contains structure events and
circle events that are later canceled, e.g., when a site is detected
within the corresponding witness circle.

For two sites $S, T \in \SitesInSweep$, we say that a point
$P \in \Intersection{\Ellipse{S}}{\Ellipse{T}}$ \emph{enters the beach
  curve at radius $r$}, if there is an $\varepsilon > 0$ such that
$P \notin \Set{B}$ for $\SweepRadius \in [r - \varepsilon, r)$ and
$P \in \Set{B}$ when $\SweepRadius = r$.  Analogously, we say that
$P \in \Intersection{\Ellipse{S}}{\Ellipse{T}}$ \emph{leaves the beach
  curve at radius} $r$ if $P \in \Set{B}$ when $\SweepRadius = r$ and
there is an $\varepsilon > 0$ such that $P \notin \Set{B}$ for
$\SweepRadius \in (r, r + \varepsilon]$.  In the following, we show
that no beach ellipse intersection enters or leaves the beach curve
between two events in $\Set{\bar{Q}}$.

\begin{lemma}
  \label{lem:intersections-enter}
  Let $\Set{\bar{Q}}$ be the tuple of site events and true circle
  events and let $r_i, r_j \in \Set{\bar{Q}}$ be two consecutive
  events.  Then, for all $\SweepRadius \in (r_i, r_j)$ no beach
  ellipse intersection enters the beach curve.
\end{lemma}
\begin{proof}
  Let $S, T$ be two sites and let
  $I \in \Intersection{\Ellipse{S}}{\Ellipse{T}}$ be an intersection
  of their ellipses.  For the sake of contradiction, assume that there
  exists an $r \in (r_i, r_j)$ such that $I$ enters the beach curve at
  $r$.  Note that no site event occurs in $(r_i, r_j)$, since this
  would contradict the construction of $\Set{\bar{Q}}$.  Consequently,
  if $I$ is on the beach curve at $\SweepRadius = r$, we know that $I$
  is contained in the sweep circle for all
  $\SweepRadius \in (r_i, r_j)$.  By
  Lemma~\ref{lem:active-means-closer} we know that $I$ not being on
  the beach curve before $\SweepRadius = r$ implies the existence of
  at least one site $S' \neq S, T$, such that
  $\Dist{S'}{I} < \Dist{S}{I} = \Dist{T}{I}$.  In particular, we
  choose $S'$ to be the one that remains active the longest at angular
  coordinate $\varphi(I)$ as the sweep circle expands beyond $r_i$.
  Since $I$ is on the beach curve at $\SweepRadius = r$, we also know
  that $\Dist{S'}{I} \ge \Dist{S}{I} = \Dist{T}{I}$ at that moment
  (again Lemma~\ref{lem:active-means-closer}).  Thus, as $I$ moves
  along the bisector $\Bisector{S}{T}$, there exists an
  $r' \in (r_i, r]$ such that
  $\Dist{S'}{I} = \Dist{S}{I} = \Dist{T}{I}$
  (Lemma~\ref{lem:distance-continuous}).  Then, $I$ is the center of
  an empty circle (as otherwise there would be yet another site that
  is longer active than $S'$, contradicting the choice of $S'$) and
  thus lies on a Voronoi vertex $V \in \Set{V}$.  However, this would
  imply that $r' \in (r_i, r_j)$ is the radius of the far point of
  $V$, which again contradicts the construction of $\Set{\bar{Q}}$.
\end{proof}

Note that if no intersections enter the beach curve between events in
$\Set{\bar{Q}}$, then also no beach ellipse segments can become active
then.  Moreover, no segments become active during a circle event
either, as only two intersections are merged into one there.  We can
conclude the following lemma, which is the hyperbolic sweep circle
counterpart of~\cite[Lemma 7.6]{bcko-cg-08} in the Euclidean sweep
line version.

\begin{lemma}
  \label{lem:segments-appear-through-sites}
  The only way in which a new beach ellipse segment can become active
  is through a site event.
\end{lemma}

We continue by investigating how beach ellipse segments disappear from
the beach curve.  Analogous to the proof of
Lemma~\ref{lem:intersections-enter}, we can prove that no beach
ellipse intersection leaves the beach curve between two events in
$\Set{\bar{Q}}$.

\begin{lemma}
  \label{lem:intersections-leave}
  Let $\Set{\bar{Q}}$ be the tuple of site events and true circle
  events and let $r_i, r_j \in \Set{\bar{Q}}$ be two consecutive
  events.  Then, for all $\SweepRadius \in (r_i, r_j)$ no beach
  ellipse intersection leaves the beach~curve.
\end{lemma}
\begin{proof}
  Consider two sites $S, T$, let
  $I \in \Intersection{\Ellipse{S}}{\Ellipse{T}}$ be an intersection
  of their ellipses, and assume for the sake of contradiction that
  there exists an $r \in (r_i, r_j)$ such that $I$ leaves the beach
  curve at $r$.  That is, there exists an $\varepsilon > 0$, such that
  $I$ is on the beach curve until $\SweepRadius = r$ and is no longer
  on the beach curve for $\SweepRadius \in (r, r + \varepsilon)$.  By
  Lemma~\ref{lem:active-means-closer}, it follows that for all sites
  $S' \neq S, T$ we have $\Dist{S'}{I} \ge \Dist{S}{I} = \Dist{T}{I}$
  at $\SweepRadius = r$, but at least one of them is closer to $I$
  than $\Dist{S}{I} = \Dist{T}{I}$ afterwards.  Let $S''$ be the one
  for which this happens first.  That is, for this site we have
  $\Dist{S''}{I} \ge \Dist{S}{I} = \Dist{T}{I}$ at $\SweepRadius = r$
  and $\Dist{S''}{I} < \Dist{S}{I} = \Dist{S}{T}$ afterwards.  By
  Lemma~\ref{lem:distance-continuous}, we know that when $I$ moves
  along the bisector $\Bisector{S}{T}$ as the sweep circle expands,
  there exists a radius $r' \in [r, r + \varepsilon)$ such that
  $\Dist{S''}{I} = \Dist{S}{I} = \Dist{T}{I}$ when
  $\SweepRadius = r'$.  As by the choice of $S''$, no other site is
  closer to $I$ than $S'', S$, and $T$, we know that $I$ is the center
  of an empty circle and thus lies on a Voronoi vertex
  $V \in \Set{V}$.  However, this implies that $r' \in (r_i, r_j)$ is
  the radius of the far point of $V$, which contradicts the
  construction of $\Set{\bar{Q}}$.
\end{proof}

Again, note that if no intersections can leave the beach curve between
events in $\Set{\bar{Q}}$, then also no beach ellipse segment can
become inactive then.  Moreover, no segments become inactive during a
site event, since we only insert two intersections consecutively on
the beach curve there.  As a result, we obtain the following lemma,
which is the hyperbolic equivalent of~\cite[Lemma 7.7]{bcko-cg-08} in
the Euclidean sweep line approach.

\begin{lemma}
  \label{lem:segments-disappear-through-circles}
  The only way in which an active beach ellipse segment can become
  inactive is through a circle event.
\end{lemma}

\subsubsection{Voronoi Vertices}

It remains to prove that all Voronoi vertices are actually found by
means of circle events.  The following lemma is the analog
of~\cite[Lemma 7.8]{bcko-cg-08} in the Euclidean version.  To simplify
the following proof, we assume that at most three sites are incident
to a Voronoi vertex $V$.  If this does not hold, the algorithm may
produce duplicate Voronoi vertices, which need to be merged in a
post-processing step.

\begin{lemma}
  \label{lem:vertices-through-circle-events}
  Every Voronoi vertex is detected by means of a circle event.
\end{lemma}
\begin{proof}
  Let $V \in \Set{V}$ be a Voronoi vertex, let $V'$ be its far point,
  and let $(S_1, S_2, S_3)$ be the incidence tuple of $V$.  Further,
  let $r \in \Set{\bar{Q}}$ be the predecessor of $r(V')$ in
  $\Set{\bar{Q}}$.  We prove that for all sweep circle radii
  $\SweepRadius \in [r, r(V'))$ there are beach ellipse intersections
  $I_{12} \in \Intersection{\Ellipse{S_1}}{\Ellipse{S_2}}$ and
  $I_{23} \in \Intersection{\Ellipse{S_2}}{\Ellipse{S_3}}$ that are
  consecutive on the beach curve $\Set{B}$.  Then, it follows that the
  corresponding circle event is in $\Set{Q}$ and the Voronoi vertex
  $V$ is detected at $\SweepRadius = r(V')$.

  We start by showing that $I_{12}, I_{23} \in \Set{B}$ for
  $\SweepRadius \in [r, r(V'))$.  In particular, we give the proof for
  $I_{12}$, as the one for $I_{23}$ is analogous.  For the sake of
  contradiction, assume that there exists an $r' \in [r, r(V'))$ such
  that $I_{12} \notin \Set{B}$ when $\SweepRadius = r'$.  By
  Lemma~\ref{lem:ellipse-intersections}, $I_{12}$ exists, so the only
  way for it \emph{not} to be on the beach curve is that there exists
  another site $S \neq S_1, S_2$ such that
  $\Dist{S}{I_{12}} < \Dist{S_1}{I_{12}} = \Dist{S_2}{I_{12}}$
  (Lemma~\ref{lem:active-means-closer}).  Moreover, by
  Corollary~\ref{cor:third-site-distinct} we have
  $\Dist{S_3}{I_{12}} > \Dist{S_1}{I_{12}} = \Dist{S_1}{I_{12}}$.  It
  follows that $S \neq S_3$.  We now show that the site $S$ that is
  distinct from $S_1, S_2$, and $S_3$ cannot exist.

  By Lemma~\ref{lem:intersections-enter} we know that $I_{12}$ does
  not enter the beach curve until at least $\SweepRadius = r(V')$.
  Now first consider the case where $I_{12}$ enters the beach curve
  exactly at $\SweepRadius = r(V')$, i.e., exactly when $I_{12} = V$.
  This means that
  $\Dist{S}{I_{12}} \ge \Dist{S_1}{I_{12}} = \Dist{S_2}{I_{12}}$ at
  this point (Lemma~\ref{lem:active-means-closer}).  In particular, we
  have $\Dist{S}{I_{12}} = \Dist{S_1}{I_{12}} = \Dist{S_2}{I_{12}}$
  then, as by Lemma~\ref{lem:distance-continuous} there exists a point
  at which we have equality but this point does not occur before
  $\SweepRadius = r(V')$.  It follows that besides $S_1, S_2$, and
  $S_3$, the site $S$ is on the witness circle of $V$, contradicting
  our assumption that no more than three sites do.  Now consider the
  case where $I_{12}$ does not enter the beach curve at
  $\SweepRadius = r(V')$, i.e., when $I_{12} = V$.  By
  Lemma~\ref{lem:active-means-closer}, we know that
  $\Dist{S}{I_{12}} < \Dist{S_1}{I_{12}}$, i.e,
  $\Dist{S}{V} < \Dist{S_1}{V}$, which means that $S$ is contained in
  the witness circle of $V$, contradicting the assumption that $V$ is
  a Voronoi vertex.  Since both cases lead to a contradiction, we can
  conclude that $I_{12} \in \Set{B}$ for all
  $\SweepRadius \in [r, r(V'))$.

  It remains to show that $I_{12}$ and $I_{23}$ are also consecutive
  on $\Set{B}$ for $\SweepRadius \in [r, r(V'))$.  First note that any
  intersection that were to lie between $I_{12}$ and $I_{23}$ cannot
  span beyond these two intersections, as this would contradict the
  fact that $I_{12}, I_{23} \in \Set{B}$, which we just proved.  It
  follows that if $I_{12}, I_{23}$ are not consecutive on $\Set{B}$,
  then there are at least two intersections $I$ and $I'$ between them
  that belong to the same active segment $\Set{A}$, which is part of
  the beach ellipse of another site $S$.  Since none of the
  intersections $I_{12}, I_{23}, I$, and $I'$ leave $\Set{B}$ until at
  least $\SweepRadius = r(V')$ (Lemma~\ref{lem:intersections-leave}),
  it follows that $I$ and $I'$ stay between $I_{12}$ and $I_{23}$
  until $\SweepRadius = r(V')$, which is when $I_{12}$ and $I_{23}$
  meet at $V$.  Then, $V = I_{12} = I_{23} = I = I'$.  Now note that
  only three intersections of the beach ellipses of the sites
  $S_1, S_2$, and $S_3$ meet at $V$, since the two beach ellipse
  intersections of a pair of them travel in opposite directions on the
  perpendicular bisector, only one of which leads to $V$
  (Lemma~\ref{lem:intersections-move}).  Thus, at least one of $I$ and
  $I'$ belongs to a beach ellipse $\Ellipse{S}$ with
  $S \neq S_1, S_2, S_3$.  It follows that a fourth site lies on the
  witness circle of $V$, contradicting the assumption that no more
  than three do.
\end{proof}

\subsection{Complexity}

To conclude the proof of Theorem~\ref{thm:main}, it remains to show
that the algorithm takes time $\mathcal{O}(n \log(n))$ to compute the
Voronoi diagram of $n$ sites.  Initially, all site events need to be
scheduled, meaning the sites have to be sorted by their radii, which
takes time $\mathcal{O}(n \log(n))$.  The running time of the
remainder of the algorithm then depends on the complexity of the
diagram, i.e., the number of Voronoi vertices and edges.  It was
previously shown that this complexity is $\mathcal{O}(n)$, by
examining different models: the Poincare disk model ~\cite[Consequence
of Proposition 2]{Bogdanov_Devillers_Teillaud_2014}, the Poincaré
half-plane model~\cite[Theorem 18.5.1]{by-ne-98} and the Klein disk
model~\cite[Theorem 1]{Nielsen_Nock_2010}.  Of course, it is no
surprise that all came to the same conclusion, since the different
models represent different ways to address points in the same space.

Clearly, there are exactly $n$ site events and the number of true
circle events is bounded by the number of Voronoi vertices, which is
$\mathcal{O}(n)$.  As each event is processed, at most a constant
number of circle events are scheduled and as the algorithm proceeds
the number of canceled events cannot be larger than the scheduled
ones.  Moreover, since each intersection can contribute at most one
structure event, it follows that the total number of processed events
is $\mathcal{O}(n)$.

It remains to show that we can handle an event in time
$\mathcal{O}(\log(n))$.  Since the queue contains $\mathcal{O}(n)$
events, inserting and removing elements from the queue, takes time
$\mathcal{O}(\log(n))$.  Regarding updating the beach curve data
structure, recall that two beach ellipses intersect at most two times
(Lemma~\ref{lem:ellipse-intersections}) and note that, consequently,
at most two intersections are on the same edge of the Voronoi diagram
at all times.  It follows that the number of elements in the beach
curve is at most $\mathcal{O}(n)$ at all times, meaning insertions and
deletions take at most $\mathcal{O}(\log(n))$ time.  All other
operations, like inserting vertices to the diagram, marking bisectors
incident to the vertices, and predicting new circle and structure
events, take constant time.

\end{document}